\documentclass{article} % For LaTeX2e
\usepackage{iclr2024_conference,times}
\usepackage{times}
\usepackage{epsfig}
\usepackage{graphicx}
\usepackage{amsmath}
\usepackage{amssymb}
\usepackage{amsthm}
\usepackage{mathtools}
\usepackage{hyperref}
\usepackage{amssymb,amsmath,amsthm,enumitem}
\usepackage{mathrsfs}
\newtheorem{theorem}{Theorem}[section]
\newtheorem{thm}{Theorem}
\usepackage[T1]{fontenc}
\newtheorem{lemma}[theorem]{Lemma}

\newtheorem{corollary}{Corollary}
\usepackage{babel,blindtext}
\usepackage{amssymb}% http://ctan.org/pkg/amssymb
\usepackage{pifont}% http://ctan.org/pkg/pifont
\newcommand{\cmark}{\ding{51}}%
\newcommand{\xmark}{\ding{55}}%
\usepackage[belowskip=-10pt,aboveskip=0pt]{caption}
\usepackage{amsmath,amsfonts,bm}

\def\eqref#1{equation~\ref{#1}}
\def\1{\bm{1}}

\DeclareMathAlphabet{\mathsfit}{\encodingdefault}{\sfdefault}{m}{sl}
\SetMathAlphabet{\mathsfit}{bold}{\encodingdefault}{\sfdefault}{bx}{n}

\usepackage{hyperref}
\usepackage{url}

\title{High-resolution Multi-spectral Image Guided DEM Super-resolution using Sinkhorn Regularized Adversarial Network}

\author{Subhajit Paul, Ashutosh Gupta \\
Space Applications Centre \\
Indian Space Research Organization\\
Ahmedabad, GJ 380015, India \\
\texttt{\{subhajitpaul, ashutoshg\}@sac.isro.gov.in}
}

\iclrfinalcopy % Uncomment for camera-ready version, but NOT for submission.
\begin{document}

\maketitle

\begin{abstract}
\vspace{-0.4cm}
Digital Elevation Model (DEM) is an essential aspect in the remote sensing domain to analyze and explore different applications related to surface elevation information. In this study, we intend to address the generation of high-resolution (HR) DEMs guided by HR multi-spectral (MX) satellite imagery as prior. To promptly regulate this process, we utilize the discriminator activations as spatial attention for the MX prior, and also introduce a Densely connected Multi-Residual Block (DMRB) module to assist in efficient gradient flow. Further, we present the notion of using Sinkhorn distance with traditional GAN to improve the stability of adversarial learning. In this regard, we provide both theoretical and empirical substantiation of better performance in terms of vanishing gradient issues and numerical convergence. We demonstrate both qualitative and quantitative outcomes with available state-of-the-art methods. Based on our experiments on DEM datasets of Shuttle Radar Topographic Mission (SRTM) and Cartosat-1, we show that the proposed model performs preferably against other benchmark methods. We also generate and visualize several high-resolution DEMs covering terrains with diverse signatures to show the performance of our model.
\end{abstract}
\setlength{\abovedisplayskip}{2.5pt}
\setlength{\belowdisplayskip}{2.5pt}
\section{Introduction}
\vspace{-0.3cm}
The Digital Elevation Model (DEM) is a digital representation of any three-dimensional surface. It is immensely useful in precision satellite data processing, geographic information systems~\citep{60}, hydrological studies~\citep{70}, urban planning~\citep{50}, and many other key applications. Due to its diverse applications, the accuracy and resolution of DEM have a substantial impact in different fields of operations~\citep{3,4}. The major sources of high-resolution (HR) elevation models are terrestrial and airborne systems with restricted coverage and they also typically suffer from several issues and systematic errors~\cite{1,2}. Hence, accurate HR DEM products are expensive, as they require special acquisition and processing techniques. As an alternative, enhancing the resolution (super-resolution) of existing DEMs can be seen as the most optimal strategy to address the shortfall. \par
Research on DEM super-resolution (SR) is limited despite its significance in remote sensing applications. Generally, traditional methods like linear, and bicubic interpolation are widely used for DEM SR, but they tend to produce smoothed outputs at high-frequency regions~\citep{8}. Reconstruction-based methods like steering kernel regression (SKR)~\citep{9} or non-local means (NLM)~\citep{10}, have been proposed to tackle this, however, they still under-perform at a large magnification factor. After the first introduction of SR using Convolutional Neural Network (SRCNN)~\citep{5}, its variant D-SRCNN was proposed by \citet{20} to address the DEM super-resolution problem which attains performance gain over the traditional methods. Later, with the introduction of Genrative Adversarial Networks (GANs)~\citep{goodfellow2016deep} and its variants in SR applications like SR using GANs (SRGAN)~\citep{6}, \citet{23} proposed a DEM super-resolution model, namely D-SRGAN, and later they suggested another model based on EffecientNetV2~\citep{24} for DEM SISR. 

However, monocular depth SR using natural HR images as guide is an emerging research field in computer vision. Some of the pioneering works in this domain are  Deformable Kernel Networks (DKN) and Faster DKN (FDKN)~\citep{506}, Fast Depth map SR (FDSR)~\citep{507}, and Deep Anisotropic Diffusion Adjustment (DADA)~\citep{508}.
%which learn sparse and spatially invariant filter kernels. Later, \citet{507} exerts a high-frequency guided module to embed the guide details in the depth map. Recently, \citet{508} achieved state-of-the-art performance by adapting the concept of guided anisotropic diffusion with CNNs.
Inspired by these, we propose a DEM SR framework that effectively utilizes information from an HR multi-spectral (MX) image guide by conditioning it with a discriminative spatial self-attention. In this regard, we also propose a new adversarial learning framework, namely SIRAN (Sinkhorn Regularized Adversarial Network), as well as generate our own dataset by using realistic coarse resolution data instead of bicubic downsampled. In next section, we briefly discuss the methodology.

\begin{figure*}
    \centering
    \includegraphics[width=\textwidth]{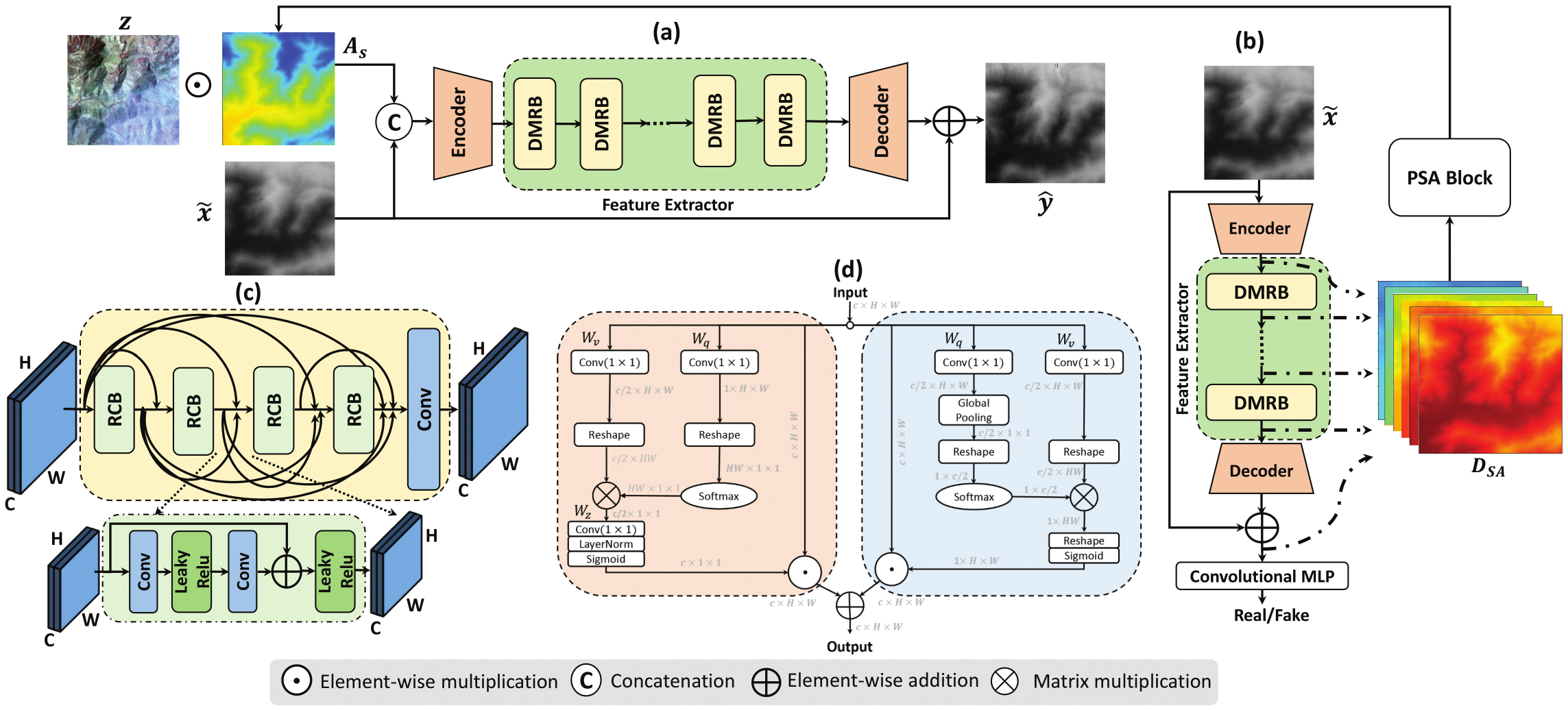}
    \vspace{0.2cm}
    \caption{Overview of the proposed adversarial framework. (a) The generator $\mathbf{G}$ takes discriminative spatial attention from (b) discriminator $\mathbf{D}$ as conditional input, paased via a (d) Polarized Self-Attention (PSA) block. Both $\mathbf{G}$ and $\mathbf{D}$ constitute of (c) Densely connected Multi Residual Blocks (DMRBs) with residual convolution block (RCB) as the building block.}
    \label{fig:fig2}
\end{figure*}
\vspace{-0.3cm}
\section{Methodology}\label{method}
\vspace{-0.3cm}
In Figure \ref{fig:fig2}, we have illustrated a detailed architectural overview of our framework. The generator $\mathbf{G}$ operates on upsampled coarser resolution DEM $\tilde{x}$, MX image prior $z$, consisting of false color composite (FCC) of $NIR$ (R), $R$ (G) and $G$ (B) bands, and polarized self-attention (PSA)~\citep{29} of discriminator spatial feature maps, $A_s$ as conditional input. Let $z \sim \mathbb{P}_{Z}$, where $z \in \mathbb{R}^{H \times W \times 3}$ with $\mathbb{P}_{Z}$ being the joint distribution of FCC composition and $\tilde{x}\sim \mathbb{P}_{\tilde{x}}$, where $\mathbb{P}_{\tilde{x}}$ constitute of upsampled coarser resolution DEM with $\tilde{x} \in \mathbb{R}^{H \times W}$. Let $\hat{y}\sim \mathbb{P}_{\mathbf{G}_\theta}$, where $\hat{y} = \mathbf{G}(\tilde{x},z\odot A_s)$, where $\odot$ denotes the element-wise multiplication and $\mathbb{P}_{\mathbf{G}_\theta}$ denotes the generator distribution parameterized by $\theta \in \Theta$. Let $y \sim \mathbb{P}_{y}$ with $\mathbb{P}_{y}$ represent the target HR DEM distribution. The discriminator $\mathbf{D}$ classifies $y$ and $\hat{y}$ to be coming from real or fake sample space and is assumed to be parameterized by $\psi \in \Psi$. 
\vspace{-0.3cm}
\subsection{Network Architecture}
\vspace{-0.25cm}
We design both $\mathbf{G}$ and $\mathbf{D}$ models based on ResNet~\citep{25} and DenseNet~\citep{26}. By combining the idea of skip and dense connections, we design a building block, namely a Densely connected Multi-Residual Block (DMRB) for our overall framework. Each DMRB block is constituted of multiple densely connected Residual Convolution Blocks (RCBs) as shown in Figure \ref{fig:fig2} (c). DMRB enables efficient context propagation and also stable gradient flow throughout the network. We kept the overall design of discriminator $\mathbf{D}$ with a similar configuration as $\mathbf{G}$ with an encoder followed by six DMRBs and finally a decoder to unravel the generated samples properly as shown in Figure \ref{fig:fig2} (a) and (b), respectively. The discriminator also adds a Multi-Layer Perceptron (MLP) layer to map its latent features into required shape. Another reason behind the design of $\mathbf{D}$ is to extract dense discriminative latent space features as they can be viewed as spatial attention to MX guide. Since, discriminators perform binary classification for a given input, apparently, in latent space, it captures the discriminative features that will help the generator focus on salient parts of the MX guide. \citet{28} introduced this concept of transferring domain-specific latent knowledge of discriminator as spatial attention to the generator. Therefore, utilizing similar concept, $\mathbf{D}$ in our proposed framework has two functional branches which are classification and to approximate the spatial attention maps. We use upsampled coarse resolution DEM $\tilde{x}$ to estimate these attentions. This choice is motivated by the fact that unlike image-to-image translation proposed by \citet{28}, during the test phase, we do not have high-resolution samples in the target domain. For this reason, we use the concept of domain adaptation loss~\citep{30}. These attention maps are passed through a PSA~\citep{29} block to exclude redundant features while highlighting significant areas by extracting dense features in both channel and spatial dimension as shown in Figure \ref{fig:fig2} (d). The main reason behind choosing PSA is due of its capability to retain the high internal resolution compared to other self-attention modules. In the following subsection, we will briefly explain our objective formulation.
\vspace{-0.2cm}
\subsection{Formulation of objective function}\label{loss}
\vspace{-0.2cm}
Our whole framework is set-up based on adversarial learning. In this regard, WGAN and its variants~\citep{300,500} serve the purpose in most of the applications due to their prompt ability to resolve the problems of conventional GAN. However, as they are designed to solve the Kantarovich formulation of OT problems to minimise the Wasserstein distance, they suffer from curse of dimensionality due to their sample complexity of $\mathcal{O}(n^{-2/d})$~\citep{34}, given a sample size $n$ with a dimension $d$. Another key concern of utilizing these adversarial objectives is the vanishing gradient problem near the optimal point. This leads the generator to converge to a sub-optimal solution and results in a partially aligned generated distribution with respect to the true distribution.
% The overall adversarial framework demonstrated in Figure \ref{fig:fig2} is trained by optimizing the minmax objective function of a conventional GAN. The main reason behind this is Wasserstein GAN (WGAN) and its variants~\citep{300,500} are based on solving the Kantarovich formulation of optimal transport (OT) which suffers from the curse of dimensionality due to its sample complexity of $\mathcal{O}(n^{-2/d})$~\citep{34}, given a sample size $n$ with dimension $d$.
Therefore, we regularize the objective function of $\mathbf{G}$ with Sinkhorn loss~\citep{36} as defined below.
\begin{equation}\label{eq17}
\begin{aligned}
\mathcal{S}_{C,\varepsilon} = \mathcal{W}_{C,\varepsilon}(\mu_\theta, \nu) - \frac{1}{2} \mathcal{W}_{C,\varepsilon}(\mu_\theta, \mu_\theta) - \frac{1}{2} \mathcal{W}_{C,\varepsilon}(\nu, \nu),
\end{aligned}
\end{equation}
where, $\mu_{\theta} \in \mathbb{P}_{\mathbf{G}_\theta}$ and $\nu \in \mathbb{P}_{y}$ are measure of generated and true distribution with support included in a compact bounded set $\mathcal{X}, \mathcal{Y} \subset \mathbb{R}^{d}$, respectively, and $\mathcal{W}_{C,\varepsilon}$ is Entropic OT~\citep{505}. Hence, it is not only computationally efficient due to its favorable sample complexity $\mathcal{O}(n^{-1/2})$~\citet{34}, equation \ref{eq17} interpolates between OT loss and MMD loss as $\varepsilon$ varies from $0$ to $\infty$. Therefore, by properly tuning $\varepsilon$ we can leverage the concurrent advantage of non-flat geometric properties of OT loss and, high dimensional rigidity and energy distance properties of MMD loss. Apart from this, the selection of $\varepsilon$ also affects the smoothness of Sinkhorn loss (see proposed Theorem 1 in Appendix \S\ref{theorem1}) which manipulates the overall gradients of $\mathbf{G}$, resulting in better prevention of vanishing gradient problems near the optimal region (see proposed Theorem 2 in Appendix \S\ref{theorem2}), and also provides tighter iteration complexity~\citep{501} (see proposed Theorem 2 in Appendix \S\ref{theorem3}) resulting in faster convergence compared to other GAN setups. Due to these advantages, we regularize the generator loss with Sinkhorn distance and refer it as $\mathscr{L}_{OT}$ that is estimated by Sinkhorn AutoDiff Algorithm~\citep{36} utilizing  $\varepsilon$ and the Sinkhorn iterations $T$ as the major parameters. As Sinkhorn loss also minimizes the Wasserstein distance, it serves the purpose of WGAN to resolve the issues of original GAN more effectively. Hence, we stick to classical GAN objective function $(\mathscr{L}_{ADV})$ for the generator to establish adversarial learning set-up while regularized with Sinkhorn loss. The generator loss also comprises of pixel loss $(\mathscr{L}_{P})$ and SSIM loss $(\mathscr{L}_{str})$. Therefore, the overall generator loss is defined as,
\begin{equation}
\lambda_P\mathscr{L}_{P} + \lambda_{str}\mathscr{L}_{str} + \lambda_{ADV}\mathscr{L}_{ADV} + \lambda_{OT}\mathscr{L}_{OT},
\end{equation}
where $\lambda_P$, $\lambda_{str}$, $\lambda_{ADV}$ and $\lambda_{OT}$ represent the weight assigned to pixel loss, SSIM loss, adversarial loss, and Sinkorn loss respectively. Similarly, the discriminator objective function can be defined as, 
\begin{align}\label{eq7}
\min _\mathbf{D} - \mathbb{E}_{y \sim \mathbb{P}_y}[\log (\mathbf{D}(y)))]& - \mathbb{E}_{\hat{y} \sim \mathbb{P}_{\mathbf{G}_\theta}}[\log(1-\mathbf{D}(\hat{y}))] + \lambda_{DA}\mathscr{L}_{D A},    
\end{align}
where $\lambda_{DA}$ is the assigned weight for domain adaptation loss $(\mathscr{L}_{D A})$~\citep{30} which is used to enforce $\mathbf{D}$ to mimic the latent features of the HR DEM and sharpen the spatial-attention maps provided an upsampled coarse DEM. The details of all losses are described in Appendix \S \ref{loss_def}.
\vspace{-0.35cm}
\section{Experiments and results analysis}
\vspace{-0.35cm}
In this section, we describe our experimental set-up followed by qualitative and quantitative comparison of our model with bicubic as well as other learning-based state-of-the-art methods such as ESRGAN~\citep{23} and  EffecientNetV2~\citep{24}, which also includes recent baseline models for image-guided depth super-resolution like DKN and FDKN~\citep{506}, DADA~\citep{508}, and FDSR~\citep{507}.
\vspace{-0.25cm}
\subsection{Dataset and implementation details}
\vspace{-0.25cm}
Due to lack of realistic DEM SR dataset, we generate our own dataset. For practical application, we opt to use real coarse resolution SRTM DEM with a ground sampling distance (GSD) of 30m as input instead of conventional bicubic downsampled. The HR Indian DEM (GSD=10m) from Cartosat-1 stereoscopic mission is considered our ground truth while the HR MX data (GSD=1.6m) from the Cartosat-2S mission serves as the image guide for our DEM SR task. The DEMs are upsampled to the resolution of MX images using bicubic interpolation to generate a paired dataset. This also helps in increasing the training samples as well as assists the model to learn dense HR features from the prior. The dataset consists of 72,000 paired patches of size $(256,256)$, featuring diverse landscapes such as vegetation, mountains, and, water regions. We use 40,000 samples for training, 20,000 for cross-validation, and 12,000 for testing, including 10,000 patches from the Indian subcontinent and the rest outside India. As ground truth DEM data is only available for Indian regions, our model is trained on limited landscapes. To check its generalization capability, we test our model on data from the Fallbrook region, US, where Cartosat DEM data is unavailable. For them, we validate our result using available 10m DEM data of the 3D Elevation Program (3DEP)~\citep{509}.
% Every experiment is conducted under identical environments. We use $3\times 3$ convolution kernel and leaky ReLU activation except in global skip connection where $1\times 1$ kernel is used without activation. Each DMRB has 64 convolution operations. We use ADAM optimizer with a fixed learning rate of 0.0001. During adversarial training, we update the critic once every single update in the generator. We set $\lambda_{DA} = 0.1$, $\lambda_P=100$, $\lambda_{str}=1$, $\lambda_{ADV}=1$ and $\lambda_{OT}=0.01$. For estimating $\mathscr{L}_{OT}$, we set $T=10$ and $\varepsilon=0.1$. The entire framework is developed using PyTorch. 
For fair comparison, the guided SR models (DKN, FDKN, DADA and FDSR) are trained from scratch with our generated dataset with MX guide image while ESRGAN and  EffecientNetV2 are trained without any guide, as they lack provisions for guide usage. Other implementation details are explained in Appendix \S \ref{impl_det}.
\subsection{Results analysis}
\begin{figure*}
    \centering
    \includegraphics[width=\textwidth]{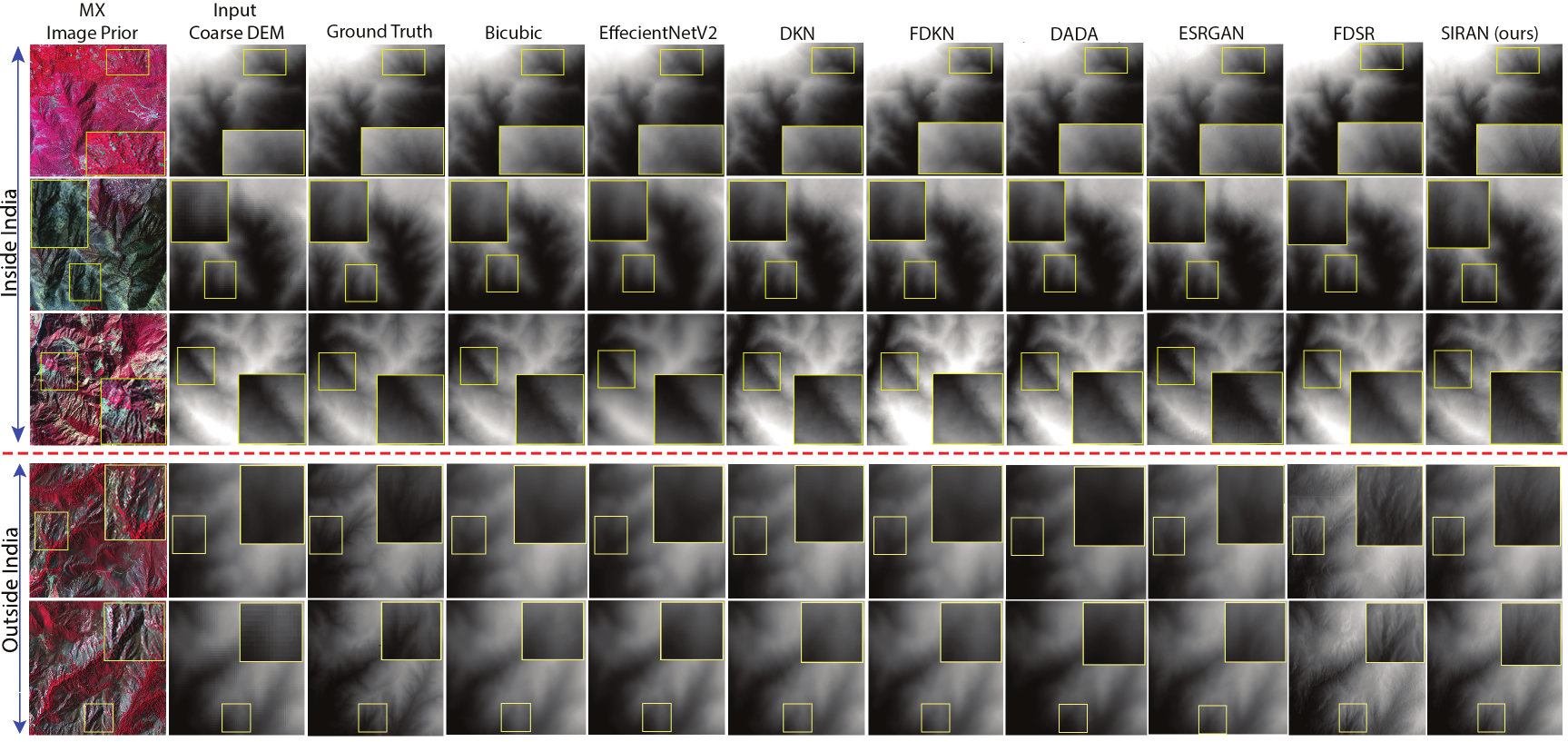}
    \vspace{0.05cm}
    \caption{Results for DEM super-resolution (better viewed at 200\%) for both inside and ouside India data and comparisons with other baseline methods.}
    % \vspace{-0.25cm}
    \label{fig:fig3}
\end{figure*}

\begin{table}[ht]
\vspace{0.25cm}
\caption{Quantitative comparison with state-of-the-art methods for both patches of inside and outside India. First and second methods are highlighted in red and green, respectively.}
\vspace{0.25cm}
\label{tab:my-table}
\resizebox{\textwidth}{!}{%
\begin{tabular}{ccccccccc}
\hline
\textbf{Method}       & \multicolumn{2}{c}{\textbf{RMSE (m)}}                                                              & \multicolumn{2}{c}{\textbf{MAE (m)}}                                                               & \multicolumn{2}{c}{\textbf{SSIM(\%)}}                                                              & \multicolumn{2}{c}{\textbf{PSNR}}                                                                  \\ \hline
\textbf{Dataset}      & \multicolumn{1}{c|}{\textbf{Inside}}                       & \textbf{Outside}                      & \multicolumn{1}{c|}{\textbf{Inside}}                       & \textbf{Outside}                      & \multicolumn{1}{c|}{\textbf{Inside}}                       & \textbf{Outside}                      & \multicolumn{1}{c|}{\textbf{Inside}}                       & \textbf{Outside}                      \\ \hline
\textbf{Bicubic}      & \multicolumn{1}{c|}{15.25}                                 & 23.19                                 & \multicolumn{1}{c|}{12.42}                                 & 22.04                                 & \multicolumn{1}{c|}{71.27}                                 & 66.49                                 & \multicolumn{1}{c|}{30.07}                                 & 27.79                                 \\
\textbf{ENetV2} \cite{24}       & \multicolumn{1}{c|}{20.35}                                 & 30.53                                 & \multicolumn{1}{c|}{18.72}                                 & 28.36                                 & \multicolumn{1}{c|}{69.63}                                 & 60.04                                 & \multicolumn{1}{c|}{31.74}                                 & 25.58                                 \\
\textbf{DKN} \cite{506}         & \multicolumn{1}{c|}{12.89}                                 & 21.16                                 & \multicolumn{1}{c|}{11.18}                                 & 19.78                                 & \multicolumn{1}{c|}{73.59}                                 & 68.45                                 & \multicolumn{1}{c|}{32.09}                                 & 28.22                                 \\
\textbf{FDKN} \cite{506}        & \multicolumn{1}{c|}{13.05}                                 & 21.93                                 & \multicolumn{1}{c|}{11.34}                                 & 20.41                                 & \multicolumn{1}{c|}{74.13}                                 & 66.83                                 & \multicolumn{1}{c|}{32.46}                                 & 27.68                                 \\
\textbf{DADA} \cite{508}        & \multicolumn{1}{c|}{37.49}                                 & 40.89                                 & \multicolumn{1}{c|}{32.17}                                 & 37.74                                 & \multicolumn{1}{c|}{73.32}                                 & 69.86                                 & \multicolumn{1}{c|}{27.94}                                 & 26.78                                 \\
\textbf{ESRGAN} \cite{23}      & \multicolumn{1}{c|}{31.33}                                 & {\color[HTML]{32CB00} \textit{20.45}} & \multicolumn{1}{c|}{25.56}                                 & {\color[HTML]{32CB00} \textit{18.34}} & \multicolumn{1}{c|}{{\color[HTML]{32CB00} \textit{82.48}}} & {\color[HTML]{32CB00} \textit{75.67}} & \multicolumn{1}{c|}{29.88}                                 & {\color[HTML]{32CB00} \textit{29.05}} \\
\textbf{FDSR} \cite{507}          & \multicolumn{1}{c|}{{\color[HTML]{32CB00} \textit{12.98}}} & 30.58                                 & \multicolumn{1}{c|}{{\color[HTML]{32CB00} \textit{10.87}}} & 25.28                                 & \multicolumn{1}{c|}{81.49}                                 & 59.81                                 & \multicolumn{1}{c|}{{\color[HTML]{32CB00} \textit{33.77}}} & 25.59                                 \\
\textbf{SIRAN (ours)} & \multicolumn{1}{c|}{{\color[HTML]{CB0000} \textbf{9.28}}}  & {\color[HTML]{CB0000} \textbf{15.74}} & \multicolumn{1}{c|}{{\color[HTML]{CB0000} \textbf{8.51}}}  & {\color[HTML]{CB0000} \textbf{12.25}} & \multicolumn{1}{c|}{{\color[HTML]{CB0000} \textbf{90.59}}} & {\color[HTML]{CB0000} \textbf{83.90}} & \multicolumn{1}{c|}{{\color[HTML]{CB0000} \textbf{35.06}}} & {\color[HTML]{CB0000} \textbf{31.56}} \\ \hline
\end{tabular}%
}
\end{table}

Figure \ref{fig:fig3} demonstrates the qualitative comparison of our proposed method, where SIRAN highlights key features and retains the perceptual quality with respect to ground truth showcasing its generalization capability. Although ESRGAN and FDSR perform well, ESRGAN tends to produce artifacts and noise while FDSR inpaints unnecessary image details in the generated HR DEMs in outside India data due to low generalization capability. The quantitative comparison presented with respect to RMSE (m), MAE (m), PSNR and SSIM in table \ref{tab:my-table} supports these observations. For both inside and outside India patches, our method captures HR structural details resulting in improvement of more than 24\% in RMSE and MAE, 8\% in SSIM, and 2dB in PSNR. Hence, both quantitative and qualitative analysis suggest that our proposed method SIRAN not only generate high-resolution DEM to the closest proximity of ground truth but also has superior generalization capability for out-of-domain samples. We further validate this by line-profile comparisons as shown in Figure \ref{fig:line} where we show SIRAN's low bias and close adherence to true elevation values. Additionally, in Figure \ref{fig:3D}, we demonstrate 3-D visualization of generated DEMs, corresponding to a region, where ground truth is unavailable. We compare it with available SRTM DEM, and clearly, our topographic view of generated DEM captures sharper features in mountainous regions as well as in the tributaries of the water basin area. Ablation studies related choice of different modules and loss functions are carried out in Appendix \S\ref{abl}.
\begin{figure}
\centering
\minipage{0.48\textwidth}
\centering
  \includegraphics[width=\linewidth]{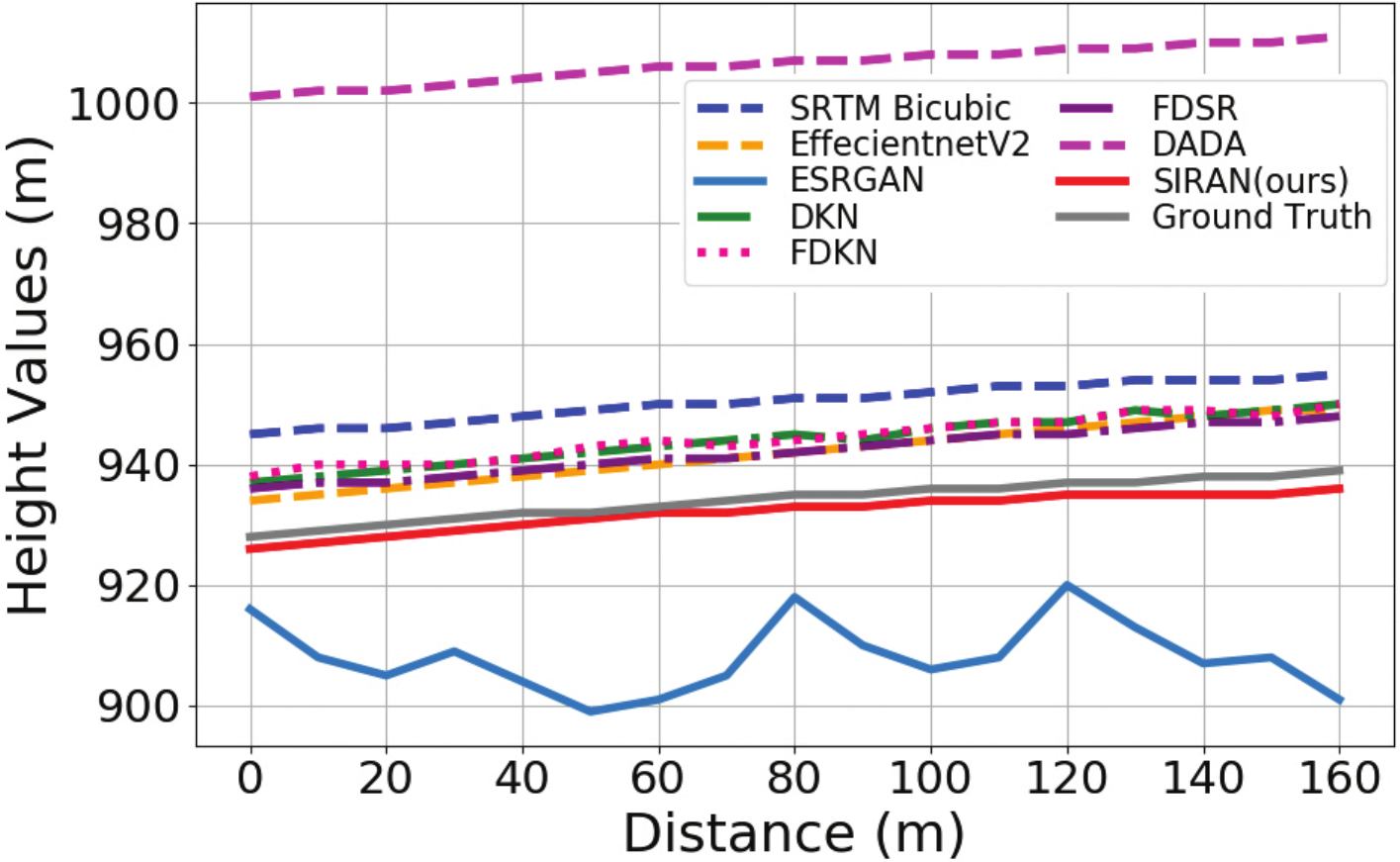}
  \vspace{0.25cm}
  \caption{Line profile analysis of SIRAN and other baselines.}\label{fig:line}
\endminipage 
\hfill
\minipage{0.48\textwidth}
\centering
  \includegraphics[width=0.85\linewidth]{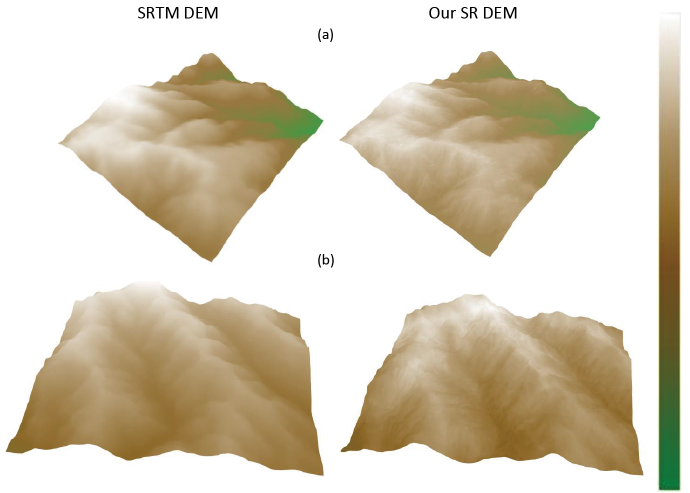}
  \vspace{0.25cm}
  \caption{3-D visualization of Super-resolved
and SRTM DEM (better at 200\%)}\label{fig:3D}
\endminipage
\end{figure}
\section{Conclusion}
% \vspace{-0.35cm}
In this paper, we demonstrate an effective approach for DEM super-resolution using realistic coarse data samples. Unlike regular SISR, the proposed method uses high-resolution MX images as prior in a specially designed architectural framework consisting of spatial attention maps from discriminator, PSA, and DMRBs. We also develop a new GAN set-up based on optimization of Sinkhorn distance regularized adversarial learning. We provide theoretical and empirical evidence to show how this choice stabilizes the training of our adversarial model and improves its convergence. We perform quantitative and qualitative analysis by generating and comparing DEMs related to different signatures along with investigating generalization capability by testing out-of-domain samples. We also analyse how each proposed module affects the outcomes. Our method achieves favorable results compared to other state-of-the-art methods.\par
\textbf{Limitation:} The generated dataset contains samples from only Indian subcontinent regions with limited landscapes and signatures. Therefore, models trained on this dataset may suffer from lack of generalization in presence of certain signatures like urban area. \par
\textbf{Dataset:} A subset of our generated dataset will be available at: \url{https://github.com/subhaISRO/DEM-Super-resolution.git}. The details about the dataset will be provided in the corresponding \url{README.md} file.

% Here, we demonstrate an effective approach for DEM super-resolution using realistic coarse data samples with HR MX guide conditioned on spatial attention maps from discriminator. We also develop a new GAN set-up based on Sinkhorn loss with theoretical and empirical evidence to show how it stabilizes the training and improves its convergence. We perform quantitative and qualitative comparisons by generating DEMs of different signatures while also investigating generalization capability by testing out-of-domain samples. Comparatively, our method achieves favorable results. 
\section{Acknowledgement}
The authors express their sincere gratitude to Shri. Nilesh M Desai, Director, Space Application Centre (SAC) Ahmedabad, Shri. Debajyoti Dhar, Deputy Director, Signal and Image Processing Area (SIPA) and Shri. S Devakanth Naidu, Head, High-resolution Data Processing Division (HDPD) for the encouragement; and other HDPD members for their constant support. We also thank National Remote Sensing Centre (NRSC), ISRO for providing necessary datasets.

\bibliography{iclr2024_conference}
\bibliographystyle{iclr2024_conference}

\newpage
\appendix
\section{Appendix: Smoothness of Sinkhorn loss}\label{theorem1}
\begin{thm}[Smoothness of Sinkhorn loss]
\textit{Consider the Sinkhorn loss $\mathcal{S}_{C,\varepsilon}(\mu_\theta, \nu)$  between two measures $\mu_\theta$ and $\nu$ on $\mathcal{X}$ and $\mathcal{Y}$ two bounded subsets of $\mathbb{R}^{d}$, with a $\mathcal{C}^{\infty}$, $L_0$-Lipschitz, and $L_1$-smooth cost function $C$. Then, for $(\theta_1, \theta_2)\in \Theta$, one has,}
\begin{equation}\label{eq18}
\begin{aligned}
&\mathbb{E}||\nabla_\theta \mathcal{S}_{C,\varepsilon}(\mu_{\theta_1}, \nu) - \nabla_\theta \mathcal{S}_{C,\varepsilon}(\mu_{\theta_2}, \nu)|| \\
&= \mathcal{O}(L(L_1 + \frac{2L_0^2L}{\varepsilon(1+B e^{\frac{\kappa}{\varepsilon}})})) ||\theta_1 - \theta_2||, 
\end{aligned}
\end{equation}
\textit{where $L$ is the Lipschitz in $\theta$ corresponding to $\mathbf{G}$, $\kappa = 2(L_0 |\mathcal{X}| + ||C||_{\infty})$}, $B = d.\max (||m||, ||M||)$ with $m$ amd $M$ being the minimum and maximum values in set $\mathcal{X}$. Let $\Gamma_{\varepsilon}$ represent the smoothness mentioned above, then we get the following asymptotic behavior in $\varepsilon$:
\begin{enumerate}
    \item \textit{as $\varepsilon \rightarrow 0$, $\Gamma_{\varepsilon} \rightarrow \mathcal{O}(\frac{2 L_0^2 L^2}{B \varepsilon e^{\frac{\kappa}{\varepsilon}}})$} 
    \item \textit{as $\varepsilon \rightarrow \infty$, $\Gamma_{\varepsilon} \rightarrow \mathcal{O}(L L_1)$}
    \end{enumerate}
\end{thm}
\begin{proof}
    We will define some of the terminologies, which are necessary for this proof. From equation 6 of the main paper, the entropic optimal transport~\citep{505} can be defined as, 
\begin{equation} \label{sup:eq1}
\begin{aligned}
&\mathcal{W}_{C,\varepsilon}\left(\mu_{\theta},\nu\right) = 
\inf_{\pi \in \Pi\left(\mu_{\theta},\nu\right)} \int_{\mathcal{X} \times \mathcal{Y}} [C\left(\mathbf{G}_{\theta}\left(x\right),y\right)] d\pi\left(\mathbf{G}_{\theta}\left(x\right),y\right)+ \varepsilon I_{\pi}\left(\mathbf{G}_{\theta}\left(x\right),y\right),\\
& \text{where } I_{\pi}\left(\mathbf{G}_{\theta}\left(x\right),y\right)) = \int_{\mathcal{X} \times \mathcal{Y}}[\log\left(\frac{\pi\left(\mathbf{G}_{\theta}\left(x\right),y\right)}{\mu_{\theta}\left(\mathbf{G}_{\theta}\left(x\right)\right)\nu\left(y\right)}\right)]d\pi\left(\mathbf{G}_{\theta}\left(x\right),y\right),\\
 & \text{s.t.} \int_{\mathcal{X}} \pi\left(\mathbf{G}_{\theta}\left(x\right),y\right) dx = \nu\left(y\right), \: \int_{\mathcal{Y}} \pi\left(\mathbf{G}_{\theta}\left(x\right),y\right) dy = \mu_{\theta}\left(\mathbf{G}_{\theta}\left(x\right)\right) \: \& \: \pi\left(\mathbf{G}_{\theta}\left(x\right),y\right) \geq 0.
\end{aligned}
\end{equation}  
The formulation in equation \ref{sup:eq1} corresponds to the primal problem of regularized OT and, this allows us to express the dual formulation of regularized OT as the maximization of an expectation problem, as shown in equation \ref{sup:eq2}~\citep{510}.

\begin{equation}\label{sup:eq2}
\begin{aligned}
\mathcal{W}&_{C,\varepsilon}\left(\mu_{\theta},\nu\right) = 
\sup_{\phi, \psi \in \Phi}  \int_{\mathcal{X}} \phi\left(\mathbf{G}_{\theta}\left(x\right)\right) \, d\mu_{\theta}\left(\mathbf{G}_{\theta}\left(x\right)\right) + \int_{\mathcal{Y}} \psi\left(y\right) \, d\nu\left(y\right) \\ 
&- \varepsilon \int_{\mathcal{X}\times \mathcal{Y}} e^{\left(\frac{\phi\left(\mathbf{G}_{\theta}\left(x\right)\right) +  \psi\left(y\right) - C\left(\mathbf{G}_{\theta}\left(x\right),y\right)}{\varepsilon}\right)} \, d\mu_{\theta}\left(\mathbf{G}_{\theta}\left(x\right)\right)d\nu\left(y\right)  + \varepsilon\\
\end{aligned}
\end{equation}
where $\Phi = \{\left(\phi,\psi\right) \in \mathcal{C}\left(\mathcal{X}\right) \times \mathcal{C}\left(\mathcal{Y}\right)\}$ is set of real valued continuous functions for domain $\mathcal{X}$ and $\mathcal{Y}$ and they are referred as dual potentials. Now, given optimal dual potentials $\phi^{*}\left(\cdot\right)$, and $\psi^{*}\left(\cdot\right)$, the optmal coupling $\pi^{*}\left(\cdot\right)$ as per \citep{510} can be defined as
\begin{equation}\label{sup:eq3}
    \pi^{*}\left(\mathbf{G}_{\theta}\left(x\right),y\right) = \mu_{\theta}\left(\mathbf{G}_{\theta}\left(x\right)\right) \nu\left(y\right) e^{\frac{\phi^{*}\left(\mathbf{G}_{\theta}\left(x\right)\right) + \psi^{*}\left(y\right) - C\left(\mathbf{G}_{\theta}\left(x\right),y\right)}{\varepsilon}}.
\end{equation}
 To prove \textbf{Theorem 1}, we need an important property regarding its Lipschitz continuity of the dual potentials, which is explained in the following \textbf{Lemma}.
\begin{lemma}
   \textit{If $C\left(\cdot\right)$ is $L_0$ Lipschitz, then the dual potentials are also $L_0$ Lipschitz.}
\end{lemma}
\begin{proof}
    Assuming $\hat{y} = \mathbf{G}_{\theta}\left(x\right)$, then $C\left(\hat{y},y\right)$ is $L_0$-Lipschitz in $\hat{y}$. As, the entropy $I_{\pi}\left(\cdot\right)$ is selected as Shannon entropy, according to \citet{33} using the softmin operator, the optimal potential $\phi^{*}\left(\cdot\right)$ satisfy the following equation
    \begin{equation}\label{sup:eq4}
        \phi^{*}\left(\hat{y}\right) = -\varepsilon \ln{\left[\int_{\mathcal{Y}} \exp{\left(\frac{\psi^{*}\left(y\right)-C\left(\hat{y},y\right)}{\varepsilon}\right)}dy\right]}
    \end{equation}
    Now, to estimate the Lipschitz of $\phi^{*}$, we have to find the upper bound of $||\nabla_{\hat{y}} \phi^{*}\left(\hat{y}\right)||$. Hence, taking the gradient of equation \ref{sup:eq4} with respect to $\hat{y}$, the upper-bound of its norm can be written as,
  
    \begin{align}\label{sup:eq5}
         ||\nabla_{\hat{y}} \phi^{*}\left(\hat{y}\right)|| = \frac{||\int_{\mathcal{Y}} \exp{\left(\frac{\psi^{*}\left(y\right)-C\left(\hat{y},y\right)}{\varepsilon}\right)}\nabla_{\hat{y}} C\left(\hat{y},y\right)dy||}{||\int_{\mathcal{Y}} \exp{\left(\frac{\psi^{*}\left(y\right)-C\left(\hat{y},y\right)}{\varepsilon}\right)}dy||}
    \end{align}
       Now due to Lipschitz continuity of $C\left(\hat{y},y\right)$, we can say $\nabla_{\hat{y}} ||C\left(\hat{y},y\right)||\leq L_0$. Hence, using Cauchy-Schwarz inequality we will get,
       \begin{align}\label{sup:eq6}
         ||\nabla_{\hat{y}} \phi^{*}\left(\hat{y}\right)|| \leq ||\nabla_{\hat{y}} C\left(\hat{y},y\right)||
         \frac{||\int_{\mathcal{Y}} \exp{\left(\frac{\psi^{*}\left(y\right)-C\left(\hat{y},y\right)}{\varepsilon}\right)dy||}}{||\int_{\mathcal{Y}} \exp{\left(\frac{\psi^{*}\left(y\right)-C\left(\hat{y},y\right)}{\varepsilon}\right)}dy||} = L_0.
    \end{align}
This completes the proof of the lemma. An alternative proof is provided by \citet{511} in Proposition 4. Similarly, it can be proved for the other potential term.
\end{proof}
For any $\theta_1$, $\theta_2\in \Theta$ will result in different coupling solutions $\pi^{*}_{i}$, for $i=1,2$. Now, based on Danskins' theorem for optimal coupling $\pi^{*}\left(\theta\right)$, we can write
\begin{align}\label{sup:eq7}
    \nabla_{\theta} \mathcal{W}_{C,\varepsilon}\left(\mu_{\theta},\nu\right) = \mathbb{E}_{\mathbf{G}_{\theta}\left(x\right),y \sim \pi^{*}\left(\theta\right)}\left[\nabla_\theta C\left(\mathbf{G}_{\theta}\left(x\right),y\right)\right]
\end{align}
Therefore, for any $\theta_1$ and $\theta_2$, we can write,
\begin{equation}
\begin{aligned}\label{sup:eq77}
    &||\nabla_{\theta} \mathcal{W}_{C,\varepsilon}\left(\mu_{\theta_1},\nu\right) - \nabla_{\theta} \mathcal{W}_{C,\varepsilon}\left(\mu_{\theta_2},\nu\right)|| \leq \\
    & ||\mathbb{E}_{\mathbf{G}_{\theta_1}\left(x\right),y \sim \pi^{*}_1}\left[\nabla_\theta C\left(\mathbf{G}_{\theta_1}\left(x\right),y\right)\right] - \mathbb{E}_{\mathbf{G}_{\theta_1}\left(x\right),y \sim \pi^{*}_2}\left[\nabla_\theta C\left(\mathbf{G}_{\theta_1}\left(x\right),y\right)\right]|| \\
    &+ ||\mathbb{E}_{\mathbf{G}_{\theta_1}\left(x\right),y \sim \pi^{*}_2}\left[\nabla_\theta C\left(\mathbf{G}_{\theta_1}\left(x\right),y\right)\right] - \mathbb{E}_{\mathbf{G}_{\theta_2}\left(x\right),y \sim \pi^{*}_2}\left[\nabla_\theta C\left(\mathbf{G}_{\theta_2}\left(x\right),y\right)\right]||\\
    & \leq L_0L||\pi^{*}_1 - \pi^{*}_2|| + L_1L||\theta_1 - \theta_2||
\end{aligned}
\end{equation}
Now with respect to different $\theta_i$, for $i=1,2$ with different pair of dual potentials, the $||\pi^{*}_1 - \pi^{*}_2||$ can be written as below. For simplicity we denote $\mu_\theta \equiv \mu_{\theta}\left(\mathbf{G}_{\theta}\left(x\right)\right)$ and $\nu \equiv \nu\left(y\right)$.
\begin{equation}\label{sup:eq8}
    \begin{aligned}
    ||\pi^{*}_1 - \pi^{*}_2|| &= ||\mu_{\theta_1}\nu \exp{\left(\frac{\phi^{*}\left(\mathbf{G}_{\theta_1}\left(x\right)\right) + \psi^{*}\left(y\right) - C\left(\mathbf{G}_{\theta_1}\left(x\right),y\right)}{\varepsilon}\right)} \\
    &- \mu_{\theta_2}\nu \exp{\left(\frac{\phi^{*}\left(\mathbf{G}_{\theta_2}\left(x\right)\right) + \psi^{*}\left(y\right) - C\left(\mathbf{G}_{\theta_2}\left(x\right),y\right)}{\varepsilon}\right)}||  \\
  &\leq ||\nu \exp{\left(\frac{\phi^{*}\left(\mathbf{G}_{\theta_1}\left(x\right)\right) + \psi^{*}\left(y\right) - C\left(\mathbf{G}_{\theta_1}\left(x\right),y\right)}{\varepsilon}\right)}\left(\mu_{\theta_1}-\mu_{\theta_2}\right)|| \\
  & + ||\mu_{\theta_2}\nu \left[\exp{\left(\frac{\phi^{*}\left(\mathbf{G}_{\theta_1}\left(x\right)\right) + \psi^{*}\left(y\right) - C\left(\mathbf{G}_{\theta_1}\left(x\right),y\right)}{\varepsilon}\right)}\right.\\
  &-\left.\exp{\left(\frac{\phi^{*}\left(\mathbf{G}_{\theta_2}\left(x\right)\right) + \psi^{*}\left(y\right) - C\left(\mathbf{G}_{\theta_2}\left(x\right),y\right)}{\varepsilon}\right)} \right]|| 
\end{aligned}
\end{equation}
From \citet{34}, we know, as the dual potentials are $L_0$-Lipschitz, $\forall \mathbf{G}_{\theta}\left(x\right)\in \mathcal{X}$, we can write, $\phi^{*}\left(\mathbf{G}_{\theta}\left(x\right)\right)\leq L_0|\mathbf{G}_{\theta}\left(x\right)|$. And from property of c-transform, for $\forall y\in \mathcal{Y}$ we can also write $\psi^{*}\left(y\right)\leq \max_{\mathbf{G}_{\theta}\left(x\right)} \phi^{*}\left(\mathbf{G}_{\theta}\left(x\right)\right) - C\left(\mathbf{G}_{\theta}\left(x\right),y\right)$. We assume $\mathcal{X}$ to be a bounded set in our case, hence, denoting $|\mathcal{X}|$ as the diameter of the space, at optimality, we can get that $\forall \mathbf{G}_{\theta}\left(x\right)\in \mathcal{X}, \: y\in \mathcal{Y}$
\begin{equation}\label{sup:eq10}
\begin{aligned}
    &\Rightarrow \phi^{*}\left(\mathbf{G}_{\theta}\left(x\right)\right) + \psi^{*}\left(y\right) \leq 2L_0|\mathcal{X}| + ||C||_{\infty}\\
    &\Rightarrow \exp{\left(\frac{\phi^{*}\left(\mathbf{G}_{\theta}\left(x\right)\right) + \psi^{*}\left(y\right) - C\left(\mathbf{G}_{\theta}\left(x\right),y\right)}{\varepsilon}\right))} \leq \exp{\left(2\frac{L_0|\mathcal{X}| + ||C||_{\infty}}{\varepsilon}\right)}
\end{aligned}
\end{equation}
 Hence, the exponential terms in equation \ref{sup:eq8} are bounded, and we can assume it has a finite Lipschitz constant $L_{exp}$. Taking $\kappa = 2\left(L_0|\mathcal{X}| + ||C||_{\infty}\right)$, and using Cauchy-Schwarz, we can rewrite equation \ref{sup:eq8} as,
 \begin{equation}\label{sup:eq11}
\begin{aligned}
    ||\pi^{*}_1 - \pi^{*}_2|| &\leq \exp{\left(\frac{\kappa}{\varepsilon}\right)}||\nu||.||\mu_{\theta_1}-\mu_{\theta_2}|| \\
&+L_{exp}||\mu_{\theta_2}||.||\nu||.||\frac{\left(\phi^{*}\left(\mathbf{G}_{\theta_1}\left(x\right)\right)- \phi^{*}\left(\mathbf{G}_{\theta_2}\left(x\right)\right)\right)- \left(C\left(\mathbf{G}_{\theta_1}\left(x\right),y\right)-C\left(\mathbf{G}_{\theta_2}\left(x\right),y\right)\right)}{\varepsilon}||\\
&\leq \exp{\left(\frac{\kappa}{\varepsilon}\right)}||\nu||.||\mu_{\theta_1}-\mu_{\theta_2}|| + 2\frac{L_{exp}L_0L}{\varepsilon}||\mu_{\theta_2}||.||\nu||.||\theta_1-\theta_2|| 
\end{aligned}
\end{equation}
Now, as the input space $\mathcal{X}$ and output space $\mathcal{Y}$ are bounded, the corresponding measures $\mu_\theta$ and $\nu$ will also be bounded. We assume, $||\mu_\theta|| \leq \lambda_1$ and $||\nu||\leq \lambda_2$. If we apply equation \ref{sup:eq10} in equation \ref{sup:eq3}, to get the upper bound of the coupling function, we will get $||\pi^{*}_1 - \pi^{*}_2||\leq \exp{\left(\frac{\kappa}{\varepsilon}\right)}||\nu||.||\mu_{\theta_1}-\mu_{\theta_2}||$ which is less than the bound in equation \ref{sup:eq11}. Then, we can find some constant upper bound of $||\pi^{*}_1 - \pi^{*}_2||$, using the assumed bounds of measures and can write $||\pi^{*}_1 - \pi^{*}_2||\leq \exp{\left(\frac{\kappa}{\varepsilon}\right)}||\nu||.||\mu_{\theta_1}-\mu_{\theta_2}||\leq K$, such that, 
\begin{equation*}
    K \leq  \exp{\left(\frac{\kappa}{\varepsilon}\right)}||\nu||.||\mu_{\theta_1}-\mu_{\theta_2}|| + 2\frac{L_{exp}L_0L}{\varepsilon}||\mu_{\theta_2}||.||\nu||.||\theta_1-\theta_2||
\end{equation*}
Then using the marginal condition as shown in in equation \ref{sup:eq1}, we can write equation \ref{sup:eq11} as,
\begin{equation}\label{sup:eq12}
\begin{aligned}
    K & \leq \lambda_1\exp{\left(\frac{\kappa}{\varepsilon}\right)}||\int_{\mathcal{X}} \pi^{*}_1 dx - \int_{\mathcal{X}} \pi^{*}_2 dx|| + 2\lambda_1 \lambda_2 \frac{L_{exp}L_0L}{\varepsilon}||\theta_1-\theta_2||\\
    & \leq \lambda_1\exp{\left(\frac{\kappa}{\varepsilon}\right)} \int_{\mathcal{X}} ||\pi^{*}_1 - \pi^{*}_2||.|dx| + 2\lambda_1 \lambda_2 \frac{L_{exp}L_0L}{\varepsilon}||\theta_1-\theta_2||\\
    & \leq \lambda_1\exp{\left(\frac{\kappa}{\varepsilon}\right)}K \int_{\mathcal{X}} |dx| +2\lambda_1 \lambda_2 \frac{L_{exp}L_0L}{\varepsilon}||\theta_1-\theta_2|| 
\end{aligned}
\end{equation}
The input set is a compact set such that $\mathcal{X} \subset \mathbb{R}^d$. So, assuming $m$ and $M$ to be the minimum and maximum value in set $\mathcal{X}$ and considering the whole situation in discrete space, equation \ref{sup:eq12}, can be rewritten as,
\begin{equation}\label{sup:eq13}
\begin{aligned}
    K &\leq \lambda_1\exp{\left(\frac{\kappa}{\varepsilon}\right)}K \sum_{x\in \mathcal{X}} |x| + 2\lambda_1 \lambda_2 \frac{L_{exp}L_0L}{\varepsilon}||\theta_1-\theta_2||\\
    & \leq \lambda_1\exp{\left(\frac{\kappa}{\varepsilon}\right)}Kd\max\left(||M||,|||m|\right) + 2\lambda_1 \lambda_2 \frac{L_{exp}L_0L}{\varepsilon}||\theta_1-\theta_2||, 
\end{aligned}
\end{equation}
Now, taking $B = d\max\left(||M||,|||m|\right)$, and doing necessary subtraction and division on both sides of equation \ref{sup:eq13}, it can be rewritten as
\begin{equation}\label{sup:eq14}
    \begin{aligned}
        K & \leq \frac{2\lambda_1 \lambda_2L_{exp}L_0L}{\varepsilon\left(1-\lambda_1 B\exp{\left(\frac{\kappa}{\varepsilon}\right)}\right)}||\theta_1-\theta_2||\\
    & \leq \frac{2\lambda_1 \lambda_2L_{exp}L_0L}{\varepsilon\left(1+\lambda_1 B\exp{\left(\frac{\kappa}{\varepsilon}\right)}\right)}||\theta_1-\theta_2||
    \end{aligned}
\end{equation}
Equation \ref{sup:eq14}, satisfies because $\frac{\kappa}{\varepsilon}\geq0$. As, $||\pi^{*}_1 - \pi^{*}_2||\leq K$, from equation \ref{sup:eq14}, it can be written as 
\begin{equation}\label{sup:eq15}
    ||\pi^{*}_1 - \pi^{*}_2|| \leq \frac{2\lambda_1 \lambda_2L_{exp}L_0L}{\varepsilon\left(1+\lambda_1 B\exp{\left(\frac{\kappa}{\varepsilon}\right)}\right)}||\theta_1-\theta_2||
\end{equation}
Substituting equation \ref{sup:eq15} in equation \ref{sup:eq77}, we will get,
\begin{equation}\label{sup:eq16}
\begin{aligned}
    ||\nabla_{\theta} \mathcal{W}_{C,\varepsilon}\left(\mu_{\theta_1},\nu\right) - \nabla_{\theta} \mathcal{W}_{C,\varepsilon}\left(\mu_{\theta_2},\nu\right)|| &\leq L_0L||\pi^{*}_1 - \pi^{*}_2|| + L_1L||\theta_1 - \theta_2||\\
    & \leq \left(L_1L + \frac{2\lambda_1 \lambda_2L_{exp} L_0^2 L^2}{\varepsilon\left(1+\lambda_1 B\exp{\left(\frac{\kappa}{\varepsilon}\right)}\right)}\right)||\theta_1-\theta_2|| 
\end{aligned}
\end{equation}
So, the EOT problem defined in equation \ref{sup:eq1} has $\hat{\Gamma}_{\varepsilon}$ smoothness in $\theta$ with $\hat{\Gamma}_{\varepsilon} = L_1L + \frac{2\lambda_1 \lambda_2L_{exp} L_0^2 L^2}{\varepsilon\left(1+\lambda_1 B\exp{\left(\frac{\kappa}{\varepsilon}\right)}\right)}$. From this, we can derive the smoothness of Sinkhorn loss defined in equation \ref{eq17}. Note that only the first two terms in equation \ref{eq17} are $\theta$ dependent. Therefore, they only contribute to the gradient approximation and both of them will satisfy the same smoothness condition as defined in equation \ref{sup:eq16}. So, if Sinkhorn loss has smoothness $\Gamma_{\varepsilon}$, it will satisfy,
$\Gamma_{\varepsilon} = \frac{3}{2}\hat{\Gamma}_{\varepsilon}$. In general, we can define the smoothness of Sinkhorn loss with $\left(\theta_1, \theta_2\right) \in \Theta$ as,
\begin{equation}
\begin{aligned}
    ||\nabla_{\theta}S_{C,\varepsilon}\left(\mu_{\theta_1}, \nu\right) - \nabla_{\theta}S_{C,\varepsilon}\left(\mu_{\theta_2}, \nu\right)||\leq \mathcal{O} \left(L_1L + \frac{2L_0^2 L^2}{\varepsilon\left(1+B\exp{\left(\frac{\kappa}{\varepsilon}\right)}\right)}\right)||\theta_1-\theta_2||
\end{aligned}
\end{equation}
This completes the statement of \textbf{Theorem 1}
\end{proof}
\section{Appendix: Upper-bound of expected gradient in SIRAN set-up}\label{theorem2}
\begin{thm}
    \textit{Let $l(\cdot)$, $g(\cdot)$ and $\mathcal{S}_{C,\varepsilon}(\cdot)$ be the objective functions related to supervised losses, adversarial loss and Sinkhorn loss respectively, and $\theta^{*}$ and $\psi^{*}$ be the parameters of optimal generator $\mathbf{G}$ and discriminator $\mathbf{D}$. Let us suppose $l(p,y)$, where $p=\mathbf{G}_{\theta}(x)$, is $\beta$-smooth in $p$. If $||\theta - \theta^{*}||\leq \epsilon$ and $||\psi - \psi^{*}||\leq \delta$, then $||\nabla_\theta \mathbb{E}_{(x,y) \sim \mathcal{X} \times \mathcal{Y}}[l(\mathbf{G}_{\theta}(x),y) + \mathcal{S}_{C,\varepsilon}(\mu_\theta(\mathbf{G}_{\theta}(x)), \nu(y)) - g(\psi; \mathbf{G}_{\theta}(x))]||\leq L^2 \epsilon(\beta + \Gamma_{\varepsilon}) + L\delta$, where $\Gamma_{\varepsilon}$ is the derived smoothness of Sinkhorn loss in equation \ref{eq18}.}
\end{thm}
\begin{proof}
    This proof is inspired by \citet{501}. Assuming $\Gamma = \mathcal{O}\left(L_1 + \frac{2L_0^2}{\varepsilon\left(1+B\exp{\left(\frac{\kappa}{\varepsilon}\right)}\right)}\right)$ be the smoothness in $p$ for Sinkhorn loss $S_{C,\varepsilon}\left(\mu_{\theta}\left(p\right), \nu\left(y\right)\right)$, where $p=\mathbf{G}_{\theta}\left(x\right)$. For simplicity, we use a common set for inputs and outputs as $\mathcal{P}$. Hence, to approximate the gradient of Sinkhorn loss, using Jensen's inequality, we can write,
\begin{equation}
\begin{aligned}\label{sup:eq18}
    &||\nabla_{\theta} \mathbb{E}_{\left(x,y\right)\sim \mathcal{P}}[S_{C,\varepsilon}\left(\mu_{\theta}\left(\mathbf{G}_{\theta}\left(x\right)\right), \nu\left(y\right)\right)]||\\
    &\leq \mathbb{E}_{\left(x,y\right)\sim \mathcal{P}} \left[||\nabla_{\theta} S_{C,\varepsilon}\left(\mu_{\theta}\left(\mathbf{G}_{\theta}\left(x\right)\right), \nu\left(y\right)\right)||\right]\\
     & \leq \mathbb{E}_{\left(x,y\right)\sim \mathcal{P}} \left[ \underbrace{||\nabla_p S_{C,\varepsilon}\left(\mu_{\theta}\left(p\right), \nu\left(y\right)\right)|| . ||\nabla_\theta \mathbf{G}_{\theta}\left(x\right)||}_\textrm{Cauchy-Schwarz inequality}\right]\\
     & \leq L\mathbb{E}_{\left(x,y\right)\sim \mathcal{P}} \left[||\nabla_p S_{C,\varepsilon}\left(\mu_{\theta}\left(p\right), \nu\left(y\right)\right)||\right] 
\end{aligned}
\end{equation}
Say, for optimized parameter $\theta^{*}$, $t= \mathbf{G}_{\theta^{*}}(x)$. Since, $||\theta - \theta^{*}||$, we can write using the smoothness of sinkhorn loss and Lipschitz of model parameters,
\begin{equation}
\begin{aligned}\label{sup:eq19}
    &||\nabla_p S_{C,\varepsilon}\left(\mu_{\theta}\left(p\right), \nu\left(y\right)\right)|| - ||\nabla_t S_{C,\varepsilon}\left(\mu_{\theta^{*}}\left(t\right), \nu\left(y\right)\right)|| \\
    &\leq ||\nabla_p S_{C,\varepsilon}\left(\mu_{\theta}\left(p\right), \nu\left(y\right)\right) - \nabla_t S_{C,\varepsilon}\left(\mu_{\theta^{*}}\left(t\right), \nu\left(y\right)\right)||\\
    & \leq \Gamma ||p-t|| = \Gamma ||\mathbf{G}(\theta)(x) - \mathbf{G}(\theta^{*}(x))||\\
    & \leq \Gamma L ||\theta - \theta^{*}|| \leq \Gamma L \epsilon 
\end{aligned}
\end{equation}
At optimal condition, $||\nabla_t S_{C,\varepsilon}\left(\mu_{\theta^{*}}\left(t\right), \nu\left(y\right)\right)|| = 0$ as the distributions of $y$ and $t = \mathbf{G}_{\theta^*}(x)$ are aligned for optimal $\theta^*$. So, by substituting equation \ref{sup:eq19} in equation \ref{sup:eq18}, we will get
\begin{equation}
\begin{aligned}\label{sup:eq20}
    ||\nabla_{\theta} \mathbb{E}_{\left(x,y\right)\sim \mathcal{P}}\left[S_{C,\varepsilon}\left(\mu_{\theta}\left(\mathbf{G}_{\theta}\left(x\right)\right), \nu\left(y\right)\right)\right]||\leq L^2\Gamma\epsilon 
\end{aligned}
\end{equation}
From Lemma 1 of \citet{501}, we get,
\begin{equation}
\begin{aligned}\label{sup:eq21}
    ||\nabla_{\theta} \mathbb{E}_{\left(x,y\right)\sim \mathcal{P}} \left[l(\mathbf{G}_{\theta}(x),y)\right]||\leq L^2\beta\epsilon
\end{aligned}    
\end{equation}
Similarly, from Lemma 2 of \citet{501}, we get
\begin{equation}
\begin{aligned}\label{sup:eq22}
    ||-\nabla_{\theta} \mathbb{E}_{\left(x,y\right)\sim \mathcal{P}} \left[g(\psi;\mathbf{G}_{\theta}(x))\right]||\leq L\delta 
\end{aligned}    
\end{equation}
Here, $\psi$ is parameters of discriminator $\mathbf{D}$. So using equations \ref{sup:eq20}, \ref{sup:eq21}, and \ref{sup:eq22}, for the combination of losses we will get,
\begin{equation}
    \begin{aligned}
    & ||\nabla_{\theta} \mathbb{E}_{\left(x,y\right)\sim \mathcal{P}} [l(\mathbf{G}_{\theta}(x),y) + S_{C,\varepsilon}\left(\mu_{\theta}\left(\mathbf{G}_{\theta}\left(x\right)\right), \nu\left(y\right)\right) \\
    &-g(\psi;\mathbf{G}_{\theta}(x))]|| \leq ||\nabla_{\theta} \mathbb{E}_{\left(x,y\right)\sim \mathcal{P}} \left[l(\mathbf{G}_{\theta}(x),y)\right]|| \\
    &+ ||\nabla_{\theta} \mathbb{E}_{\left(x,y\right)\sim \mathcal{P}}\left[S_{C,\varepsilon}\left(\mu_{\theta}\left(\mathbf{G}_{\theta}\left(x\right)\right), \nu\left(y\right)\right)\right]|| \\
    &+  ||-\nabla_{\theta} \mathbb{E}_{\left(x,y\right)\sim \mathcal{P}} \left[g(\psi;\mathbf{G}_{\theta}(x))\right]||\\
    &\leq L^2\beta\epsilon + L^2\Gamma\epsilon + L\delta = L^2\epsilon(\beta + \Gamma) +  L\delta
\end{aligned}
\end{equation}
This completes the proof.
\end{proof}
In a regular adversarial set-ups like WGAN and their variants $\epsilon \rightarrow 0$ suggests reductions of $\delta$, which eventually results in a vanishing gradient near the optimal region as shown by \citet{501}. However, in our Sinkhorn regularized adversarial framework, the upper bound is also dependent on $\Gamma_{\varepsilon}$ as exhibited in \textbf{Theorem2}. From equation \ref{eq18}, $\Gamma_{\varepsilon}$ is exponentially variable with respect to the choice of $\varepsilon$. Hence, the selection of adequate $\varepsilon$ will have a profound impact in mitigating the vanishing gradient problem suffered by regular adversarial setup. Experimental results to support this claim is shown in \S \ref{final}.
\section{Appendix: Iteration complexiety of SIRAN}\label{theorem3}
\begin{thm}
    Suppose the supervised loss $l(\theta)$ is lower bounded by $l^{*}>-\infty$ and it is twice differentiable. For some arbitrarily small $\zeta>0$, $\eta>0$ and $\epsilon_1-$stationary point with $\epsilon_1>0$, let $||\nabla g(\psi; \mathbf{G}_{\theta}(x))||\geq \zeta$, $||\nabla\mathcal{S}_{C,\varepsilon}(\mu_\theta(\mathbf{G}_{\theta}(x)), \nu(y))||\geq \eta$ and $||\nabla l(\mathbf{G}_{\theta}(x),y)||\geq \epsilon_1$, with conditions $\delta\leq \frac{\sqrt{2\epsilon_1 \zeta}}{L}$, and $\Gamma_{\varepsilon}<\frac{\sqrt{2\epsilon_1 \eta}}{L^2\epsilon}$, then the iteration complexity in Sinkhorn regularized adversarial framework is upper bounded by $\mathcal{O}(\frac{(l(\theta_0) - l^*)\beta_1}{\epsilon_1^2+2\epsilon_1(\zeta+\eta) - L^2(\delta^2+L^2\Gamma_{\varepsilon}^2\epsilon^2)})$, assuming $||\nabla^2l(\theta)||\leq \beta_1$.
\end{thm}
\begin{proof}
    This proof also follows the steps of Theorem 3 from \cite{501}. In the sinkhorn regularized adversarial setup, the parameters $\theta$ are updated using fixed step gradient descent. They iterate as,
% \begin{equation}\label{sup:eq23}
% \begin{aligned}
%      \theta_{t+1} = \theta_t - h_t \nabla (l(\theta_t)&+S_{C,\varepsilon}\left(\mu_{\theta_t}\left(\mathbf{G}_{\theta_t}\left(x\right)\right), \nu\left(y\right)\right) \\
%      &- g(\psi;\mathbf{G}_{\theta_t}(x)))
% \end{aligned}
%  \end{equation}
\begin{equation}
 \begin{aligned}\label{sup:eq23}
     \theta_{t+1} = \theta_t - h_t \nabla (l(\theta_t)+S_{C,\varepsilon}\left(\mu_{\theta_t}\left(\mathbf{G}_{\theta_t}\left(x\right)\right), \nu\left(y\right)\right) - g(\psi;\mathbf{G}_{\theta_t}(x))).
 \end{aligned}   
\end{equation}
 For simplicity, we denote $S_{C,\varepsilon}\left(\mu_{\theta_t}\left(\mathbf{G}_{\theta_t}\left(x\right)\right), \nu\left(y\right)\right) \equiv S_{C,\varepsilon}(\mu_{\theta_t}, \nu)$. Using Taylor's expansion,
\begin{equation}
\begin{aligned}\label{sup:eq24}
    l(\theta_{t+1}) = l(\theta_t) + \nabla l(\theta_t)(\theta_{t+1} - \theta_t) 
    + \frac{1}{2}(\theta_{t+1} - \theta_t)^T \nabla^2 l(\theta_t)(\theta_{t+1} - \theta_t)
\end{aligned}    
\end{equation}
Now, substituting $\theta_{t+1} - \theta_t$ from equation \ref{sup:eq23}, and using triangle inequality and Cauchy-Schwarz inequality, equation \ref{sup:eq24} can be rewritten as,
\begin{equation}
\begin{aligned}\label{sup:eq25}
     l(\theta_{t+1}) \leq &
    l(\theta_t)\!- h_t||\nabla l(\theta_t)||^2\! - h_t||\nabla l(\theta_t)||.||\nabla S_{C,\varepsilon}(\mu_{\theta_t}, \nu)|| - h_t||\nabla l(\theta_t)||.||g(\psi;\mathbf{G}_{\theta_t}(x))|| \\
    &+ h_t^2||\nabla(l(\theta_t) + S_{C,\varepsilon}(\mu_{\theta_t}, \nu) - g(\psi;\mathbf{G}_{\theta_t}(x)))||^2 \frac{||\nabla^2 l(\theta_t)||}{2}.
    \end{aligned}    
\end{equation}
Taking into account the assumptions in \textbf{Theorem 3} and utilizing Minkowski's inequality, equation \ref{sup:eq25} can be rewritten as,
\begin{equation}
\begin{aligned}\label{sup:eq26}
    l(\theta_{t+1}) \leq &
    l(\theta_t) - h_t||\nabla l(\theta_t)||^2 - h_t||\nabla l(\theta_t)||\eta 
    -  h_t||\nabla l(\theta_t)||\zeta \\ 
    & + h_t^2 (||\nabla(l(\theta_t)||^2 + ||S_{C,\varepsilon}(\mu_{\theta_t}, \nu)||^2 
     + ||g(\psi;\mathbf{G}_{\theta_t}(x)))||^2) \frac{\beta_1}{2}.
\end{aligned}    
\end{equation}
Using $h_t = \frac{1}{\beta_1}$, from equation \ref{sup:eq26}, we can write,
\begin{equation}
\begin{aligned}
    l(\theta_{t+1}) &\leq  l(\theta_t) - \frac{h_t||\nabla l(\theta_t)||^2}{2} - h_t||\nabla l(\theta_t)||\eta -  h_t||\nabla l(\theta_t)||\zeta \\
    &+ \frac{h_t||S_{C,\varepsilon}(\mu_{\theta_t}, \nu)||^2}{2} + \frac{h_t||g(\psi;\mathbf{G}_{\theta_t}(x)))||^2}{2}\\
    & \leq l(\theta_t) - \frac{h_t \epsilon_1^2}{2} - h_t\epsilon_1 \eta -h_t\epsilon_1\zeta
    + \frac{h_tL^4\Gamma^2\epsilon^2}{2}+\frac{h_tL^2\delta^2}{2}. 
\end{aligned}    
\end{equation}
Assuming T iterations to reach this $\epsilon_1$-stationary point, then for $t\leq T$, doing telescopic sum over $t$,
\begin{equation}
\begin{aligned}
    \sum_{t=0}^{T-1} l(\theta_{t+1}) -&  l(\theta_{t})  \leq \frac{-T(\epsilon_1^2+2\epsilon_1(\zeta+\eta) - L^2(\delta^2+L^2\Gamma^2\epsilon^2))}{2\beta_1}\\
    &\Rightarrow T\leq \frac{2(l(\theta_0) - l^*)\beta_1}{(\epsilon_1^2+2\epsilon_1(\zeta+\eta) - L^2(\delta^2+L^2\Gamma^2\epsilon^2))}
\end{aligned}    
\end{equation}
Therefore, using the iteration complexity definition of \cite{501}, we obtain,
\begin{equation}
\begin{aligned}
    \sup_{\theta_0\in \{\mathbb{R}^{h\times d_x},\mathbb{R}^{d_y\times h}\},l\in \mathscr{L}} \mathcal{T}_{\epsilon_1}(A_h[l,\theta_0],l) 
    = \mathcal{O}\left(\frac{(l(\theta_0) - l^*)\beta_1}{\epsilon_1^2+2\epsilon_1(\zeta+\eta) - L^2(\delta^2+L^2\Gamma_{\varepsilon}^2\epsilon^2)}\right).
\end{aligned}
\end{equation}
This completes the proof of \textbf{Theorem 3}.
\subsection{Simplified Theorem 3}\label{appendC.1}
\begin{corollary}
    Using first order Taylor series, the upper bound in \textbf{Theorem 3} becomes $\mathcal{O}(\frac{l(\theta_0) - l^*}{\epsilon_1^2+\epsilon_1(\zeta+\eta)})$.
\end{corollary}
Using the similar arguments of \textbf{Theorem 3}, and taking first-order Taylor's approximation, we get
\begin{equation}
\begin{aligned}
    l(\theta_{t+1}) &= l(\theta_{t}) - h_t||\nabla l(\theta_t)||^2 - h_t||\nabla l(\theta_t)||.||\nabla S_{C,\varepsilon}(\mu_{\theta_t}, \nu)|| - h_t||\nabla l(\theta_t)||.||g(\psi;\mathbf{G}_{\theta_t}(x))||\\
    & \leq l(\theta_t) - h_t \epsilon_1^2 - h_t\epsilon_1 \eta -h_t\epsilon_1\zeta
\end{aligned}    
\end{equation}
Taking telescopic sum over $t$ for $t\leq T$, we get 
\begin{equation}
     \sum_{t=0}^{T-1} l(\theta_{t+1}) - l(\theta_{t}) \leq -Th_t(\epsilon_1^2+\epsilon_1(\zeta+\eta))
\end{equation}
So, using the definition of iteration complexity, we get,
\begin{equation}
\begin{aligned}
    \sup_{\theta_0\in \{\mathbb{R}^{h\times d_x},\mathbb{R}^{d_y\times h}\},l\in \mathscr{L}} \mathcal{T}_{\epsilon_1}(A_h[l,\theta_0],l) 
    = \mathcal{O}\left(\frac{l(\theta_0) - l^*}{\epsilon_1^2+\epsilon_1(\zeta+\eta)}\right)
\end{aligned}
\end{equation}
This completes the proof.
\end{proof}
Since the denominator of the derived upper bound in \textbf{Theorem 3} is greater than the one mentioned in Theorem 3 by \citet{501}, we can infer that our proposed learning framework has tighter iteration complexity compared to the regular adversarial setup. This is only true when $\Gamma_{\varepsilon}<\frac{\sqrt{2\epsilon_1 \eta}}{L^2\epsilon}$, and it can be easily satisfied by appropriate choice of $\varepsilon$. It can also be verified using a simpler setup as shown in \textbf{Corollary 1}, as it increases the convergence rate from $\mathcal{O}((\epsilon_1^2 + \epsilon_1 \zeta)^{-1})$~\citep{501} to $\mathcal{O}((\epsilon_1^2 + \epsilon_1(\zeta+\eta))^{-1})$ which effects considerably in the training iterations. This also suggests that our setup is equivalent to applying two discriminator operations without involving additional computations which helps in effectively leveraging the advantage of using multiple discriminators~\citep{502}. Experiment result  to support this proposed theorem is carried out in Appendix \S\ref{final}.

\section{Appendix: Details of used losses}\label{loss_def}
Besides classification, discriminator $\mathbf{D}$ has another major functional branch, $(\mathbf{D_{SA}})$, to approximate the spatial attention maps. For any input $m$, $\mathbf{D_{SA}}$ is used to estimate its corresponding normalized spatial feature maps, $D_{SA}: \mathbb{R}^{H \times W} \rightarrow[0,1]^{H \times W}$. Let $\mathbf{D}$ consists of $k$ DMRBs and $a_i$ be the activation maps after $i^{th}$ DMRB with $c$ output channels, such that $a_i \in \mathbb{R}^{H \times W \times c}$. Since at different depths, the discriminator focuses on different features, we select $k$ different attention maps from the layers in the latent space. Eventually, attention coefficients are calculated as in \cite{28}:
\begin{equation}\label{eq1}
\centering
\mathbf{D}_{SA}(m) =\sum_{i=1}^{k} \sum_{j=1}^{c}|a_{i j}(m)|
\end{equation}
Then the domain adaptation loss $(\mathscr{L}_{D A})$ mentioned in \S \ref{loss},is defined as,
\begin{equation}
\mathscr{L}_{D A}=\mathbb{E}_{\tilde{x} \sim \mathbb{P}_{\tilde{x}}, y \sim \mathbb{P}_{y}}\left[\left\|\mathbf{D}_{S A}(\tilde{x})-\mathbf{D}_{S A}(y)\right\|_2^2\right].
\end{equation}
where, $y$ is ground truth DEM and $\tilde{x}$ is bicubic interpolated coarse SRTM DEM as mentioned in \S \ref{method}. The pixel loss $(\mathscr{L}_{P})$ and SSIM loss $(\mathscr{L}_{str})$ and adversarial loss $(\mathscr{L}_{ADV})$ are defined as,
\begin{equation} \label{eq88}
\begin{aligned}
&\mathscr{L}_{P}=\mathbb{E}_{\tilde{x} \sim \mathbb{P}_{\tilde{x}}, z \sim \mathbb{P}_{Z}, y \sim \mathbb{P}_{y}}\left[\left\|y-\mathbf{G}(\tilde{x}, z\odot A_s(\tilde{x}))\right\|_2^2\right],\\
&\mathscr{L}_{str} = \mathbb{E}_{\tilde{x} \sim \mathbb{P}_{\tilde{x}}, z \sim \mathbb{P}_{Z}, y \sim \mathbb{P}_{y}} -\log (\mathbf{SSIM}(\mathbf{G}(\tilde{x}, z\odot A_s(\tilde{x})), y)),\\
&\mathscr{L}_{ADV} = \mathbb{E}_{\tilde{x} \sim \mathbb{P}_{\tilde{x}}, z \sim \mathbb{P}_{Z}} -\log (\mathbf{D}(\mathbf{G}(\tilde{x}, z\odot A_s(\tilde{x})))).
% & \text{where, } \hat{y} = \mathbf{G}(\tilde{x}, z\odot A_s(\tilde{x})).
\end{aligned}
\end{equation}
where, $A_s(\tilde{x}) = PSA(\mathbf{D}_{SA}(\tilde{x}))$ with PSA being polarized self-attention as discussed in \S \ref{method}
\section{Appendix: Implementation Details}\label{impl_det}
Every experiment is conducted under identical environments. We use $3\times 3$ convolution kernel and leaky ReLU activation except in global skip connection where $1\times 1$ kernel is used without activation. Each DMRB has 64 convolution operations. We use ADAM optimizer with a fixed learning rate of 0.0001. During adversarial training, we update the critic once every single update in the generator. We set $\lambda_{DA} = 0.1$, $\lambda_P=100$, $\lambda_{str}=1$, $\lambda_{ADV}=1$ and $\lambda_{OT}=0.01$. For estimating $\mathscr{L}_{OT}$, we set $T=10$ and $\varepsilon=0.1$. The entire framework is developed using PyTorch
\section{Appendix: Ablation study}\label{abl}
In this section, we analyze the efficacy of our four proposed modules, image prior, spatial attention from discriminator, PSA, and Sinkhorn loss-based adversarial learning. In the absence of any prior, it leads to outcomes similar to bicubic interpolation. This results in low RMSE and MAE, compared to other cases as shown in Table \ref{tab:my-table2}. However, that affects its SSIM score as shown in Table \ref{tab:my-table2}. Using high-resolution MX image priors resolves this issue to a certain extent by increasing SSIM by more than 6\%. Yet, this degrades the performance of the model in terms of RMSE and MAE. To improve upon this situation, we introduce spatial feature maps from the discriminator. Figure \ref{fig5} shows how individual attention maps after each DMRB prioritizes certain features at different labels. However, the mean attention weights are approximately uniform as shown in Figure \ref{fig6}. PSA handles this matter by emphasizing key features. Addressing both these attentions improves the RMSE and MAE by more than 50\% as well as enhances SSIM by almost 7\% as shown in Table \ref{tab:my-table2}. The introduction of Sinkhorn distance regularization enhances the evaluation parameters further with both MAE and RMSE below 10 m and SSIM more than 90\% as shown in Table \ref{tab:my-table2}. Apart from this, sinkhorn loss also contributed to a near 2.5X faster convergence rate for the pixel loss as shown in Figure \ref{fig7}. This also supports our argument in \textbf{Theorem 3}.
\begin{table}[]
\centering
\vspace{0.35cm}
\caption{Quantitative analysis for the effect of introducing different modules for DEM super-resolution}
\vspace{0.1cm}
\label{tab:my-table2}
\resizebox{\textwidth}{!}{%
\begin{tabular}{cccc|cccc}
\hline\hline
\textbf{\begin{tabular}[c]{@{}c@{}}Image\\ Prior\end{tabular}} &
  \textbf{\begin{tabular}[c]{@{}c@{}}Spatial\\ Attention\end{tabular}} &
  \textbf{PSA} &
  \textbf{\begin{tabular}[c]{@{}c@{}}Sinkhorn\\ loss\end{tabular}} &
  \textbf{RMSE (m)} &
  \textbf{MAE (m)} &
  \textbf{SSIM (\%)} &
  \textbf{PSNR} \\ \hline
 \xmark & \xmark & \xmark & \xmark & 16.54         & 13.63         & 72.27          & 30.25          \\
 \cmark & \xmark & \xmark & \xmark & 29.32         & 25.41         & 78.29          & 28.25          \\
 \cmark & \cmark & \xmark & \xmark & 20.76         & 18.29         & 81.68          & 31.08          \\
 \cmark & \cmark & \cmark & \xmark & 18.76         & 15.13         & 85.04          & 32.21          \\
 \cmark & \cmark & \cmark & \cmark & \textbf{9.28} & \textbf{8.51} & \textbf{90.49} & \textbf{35.06} \\ \hline \hline
\end{tabular}%
}
\end{table}
\begin{figure}[!htb]
\centering
\minipage{0.45\textwidth}
  \includegraphics[width=\linewidth]{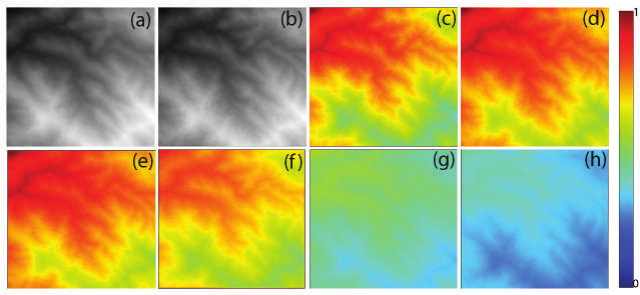}
  \caption{(a) Source, (b) Target, (c)-(h) Discriminator spatial attention after each DMRB from top to bottom.}\label{fig5}
\endminipage \quad
\minipage{0.45\textwidth}
  \includegraphics[width=\linewidth]{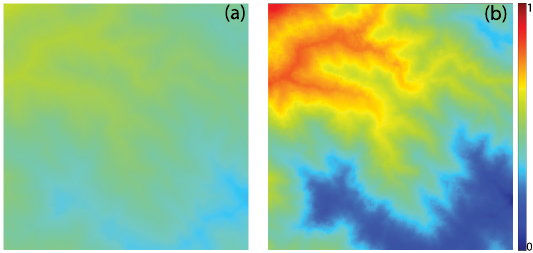}
  \caption{(a) Weights of mean discriminator spatial attention ($\textbf{D}_{SA}$), (b) weights after passing $\textbf{D}_{SA}$ through PSA block)}\label{fig6}
\endminipage \vspace{0.5cm}
\minipage{0.5\textwidth}
  \includegraphics[width=\linewidth]{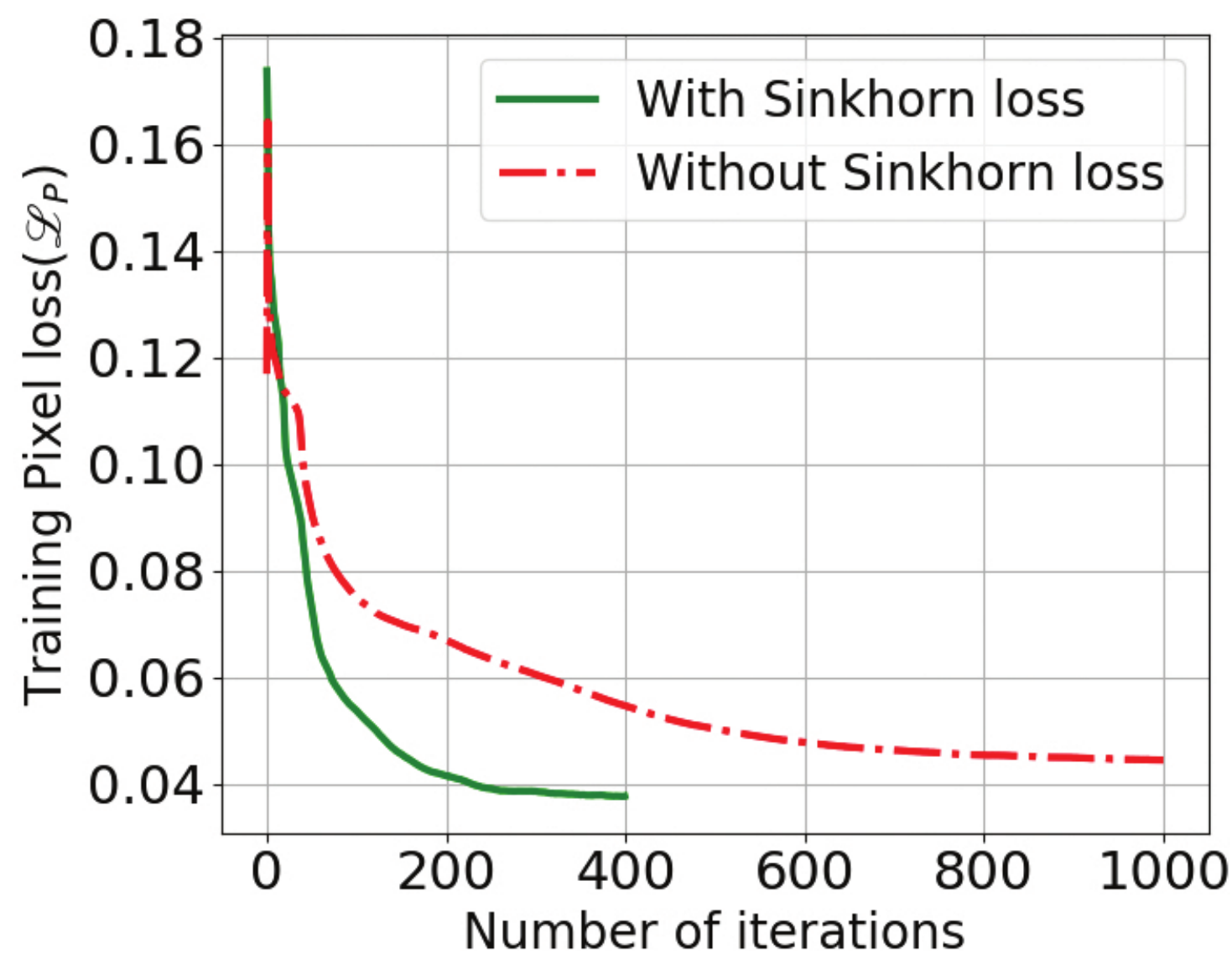}
  \caption{Effect of sinkhorn loss in training convergence for our model}\label{fig7}
  \endminipage
\end{figure}
We also perform ablation study related to different loss functions as shown in Table \ref{tab:my-table3}. This can also be visualized from \ref{fig300}. Clearly, introduction of Sinkhorn loss is pivotal in SR performance (3 dB gain). Other losses also contribute to the overall performance.
\begin{table}[]
\centering
\caption{Quantitative analysis on the effect of different losses}
\label{tab:my-table3}
\resizebox{\columnwidth}{!}{%
\begin{tabular}{cccc|cccc}
\hline
Pixel loss & SSIM loss & Adversarial Loss & Sinkhorn loss & RMSE (m) & MAE (m) & SSIM (\%) & PSNR  \\ \hline
\cmark           &  \xmark         &  \xmark                & \xmark              & 20.68    & 19.03   & 74.43     & 31.07 \\
     \cmark      &  \cmark         &  \xmark                &  \xmark             & 25.76    & 22.83   & 81.55     & 31.41 \\
     \cmark      & \cmark          & \cmark                 &   \xmark            & 18.76    & 15.13   & 85.04     & 32.21 \\
         \cmark  &  \cmark         & \cmark                 &  \cmark            & \textbf{9.28}     & \textbf{8.51}    & \textbf{90.49}     & \textbf{35.06} \\ \hline
\end{tabular}%
}
\vspace{-0.4cm}
\end{table}
\begin{figure}[]
    \centering
    \includegraphics[width=0.85\linewidth]{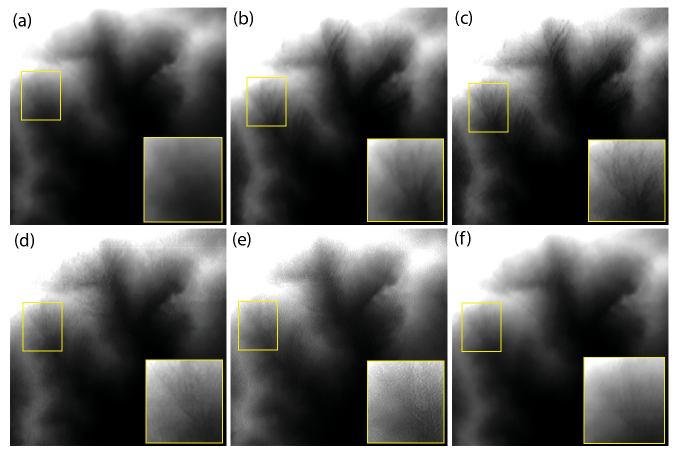}
    \caption{(a) Bicubic LR DEM, (b) HR GT DEM, predicted HR DEM corresponding to (c) all losses, (d) pixel + SSIM + adversarial loss, (e) pixel + SSIM loss, and (f) only pixel loss.}
    \label{fig300}
\end{figure}
\section{Appendix: Empirical results related to Theorem 2 and Theorem 3}\label{final}
We perform experiments to answer the proposed claims. There are two main aspects we want to investigate, firstly, how the choice of $\varepsilon$ affects the overall training of the model, and secondly, how it performs compared to other state-of-the-art learning methods like WGAN, WGAN+GP, and DCGAN. In both these cases,  we analyze the claims of mitigating vanishing gradients in the near-optimal region and fast convergence rate.
\subsection{Experiment set-up}
In this setting, we are performing a denoising operation on the MNIST dataset. For this 60000 samples of size $28\times 28$ are used during training, while 10000 are used for testing. The convergence criterion is set to be the mean square error of 0.04
or a maximum of 500 epochs. During training, we randomly add Gaussian noise to the training samples to perform the denoising task. The generator is designed as a simple autoencoder structure with an encoder and decoder each having 2 convolutional layers.
In practice, we notice that a discriminator with shallow layers is usually sufficient to offer a higher convergence rate. Therefore, we choose, a three-layer fully connected network with 1024 and 256 hidden neurons. All the layers are followed by ReLu activation except the output layer. For optimization, ADAM is utilized with a learning rate of 0.001 with a batch size of 64, and the discriminator is updated once for every single update of the generator.

\subsection{Result analysis}
Figure \ref{fig55}, shows how changing the value of $\varepsilon$ affects the overall iteration complexity. According to this figure, the instances $\varepsilon$ are very small and very large, and the learning behavior of the model becomes close to regular adversarial setup which ultimately results in more time requirement for convergence. This is because, as $\varepsilon \to 0$ and $\varepsilon \to \infty$, the smoothness of sinkhorn loss tends to become independent of $\varepsilon$ as depicted in \textbf{Theorem 1}, which makes the overall setup similar to the regular adversarial framework. This also affects the capability of mitigating the vanishing gradient problem as shown in Figure \ref{fig66} and \ref{fig77}. The gradients are approximated using spectral norm and they are moving averaged for better visualization. From Figure \ref{fig66}, in the case of the first layer, as $\varepsilon$ varies, the estimated gradients are similar near the optimal region. However, From Figure \ref{fig77}, we can see for the case of the hidden layer, gradient approximation is definitely affected by the choice of $\varepsilon$, and we can see as $\varepsilon \to 0$ and $\varepsilon \to \infty$, the gradients near-optimal region become smaller. However, using $\varepsilon=0.1$ tends to have higher gradients even if near the optimal region. Therefore, this model will have more capability of mitigating the vanishing gradient problem. Hence, we use this model to compare with other state-of-the-art learning methods.\par
We compare the rate of convergence and capability of handling the vanishing gradient of SIRAN with WGAN~\citep{300}, WGAN+GP~\citep{500}, and DCGAN. Figure \ref{fig555} clearly visualizes how our proposed framework has tighter iteration complexity than others, and reaches the convergence faster. This is consistent with the theoretical analysis presented in \textbf{Theorem 3}. Figure \ref{fig666} and \ref{fig777} also provides empirical evidence of the vanishing gradient issue presented in \textbf{Theorem 2}. Both for the first layer and hidden layer, as shown in Figure \ref{fig666} and \ref{fig777}, the approximated gradients are higher comparatively than others near the optimal region. This results in increasing the effectiveness of SIRAN in handling the issue of the vanishing gradient problem as discussed in above theorems.
\begin{figure}[!htb]
\centering
\minipage{0.32\textwidth}
  \includegraphics[width=\linewidth]{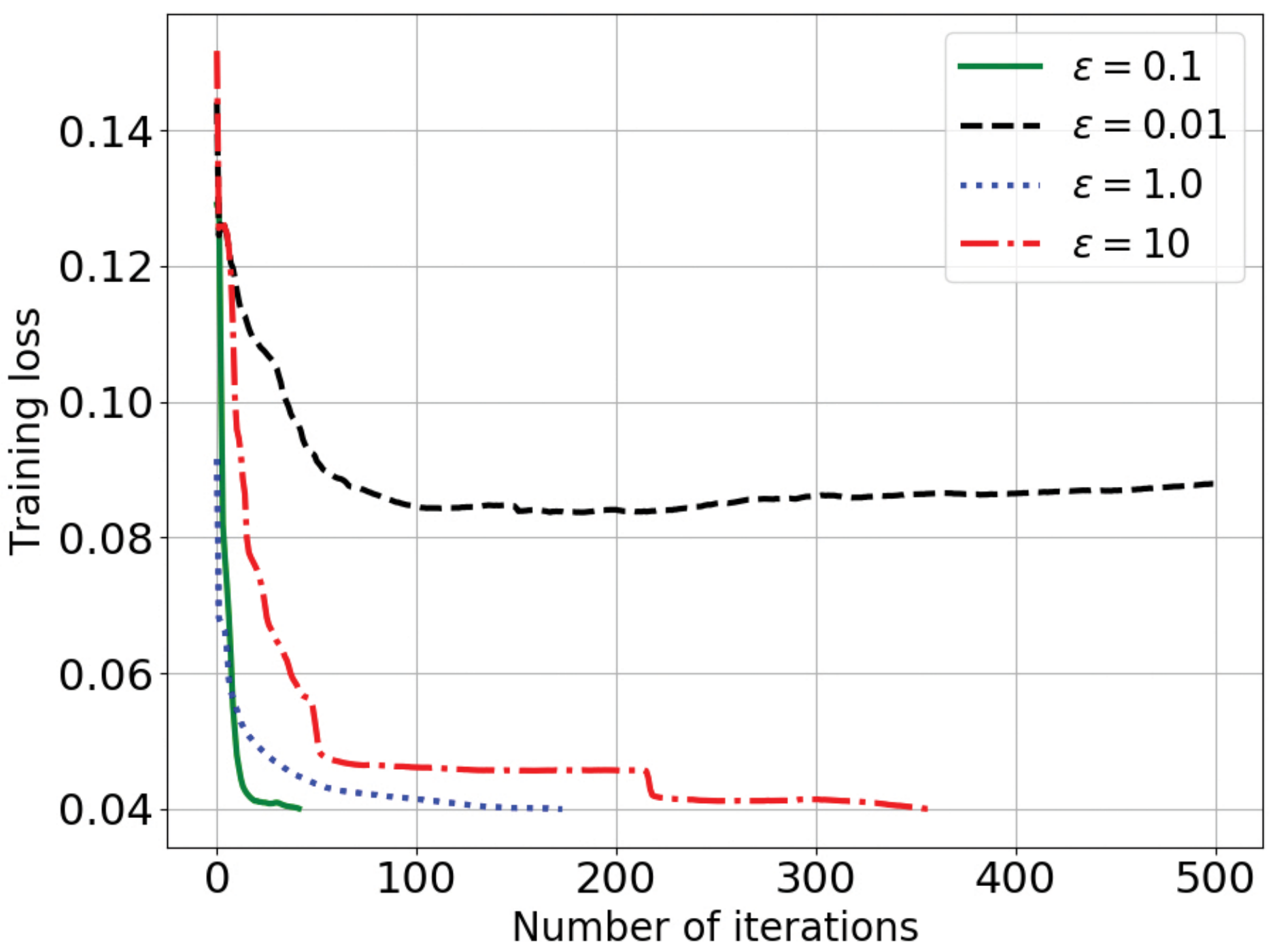}
  \caption{Training Loss for variation of $\varepsilon$}\label{fig55}
\endminipage \quad
\minipage{0.32\textwidth}
  \includegraphics[width=\linewidth]{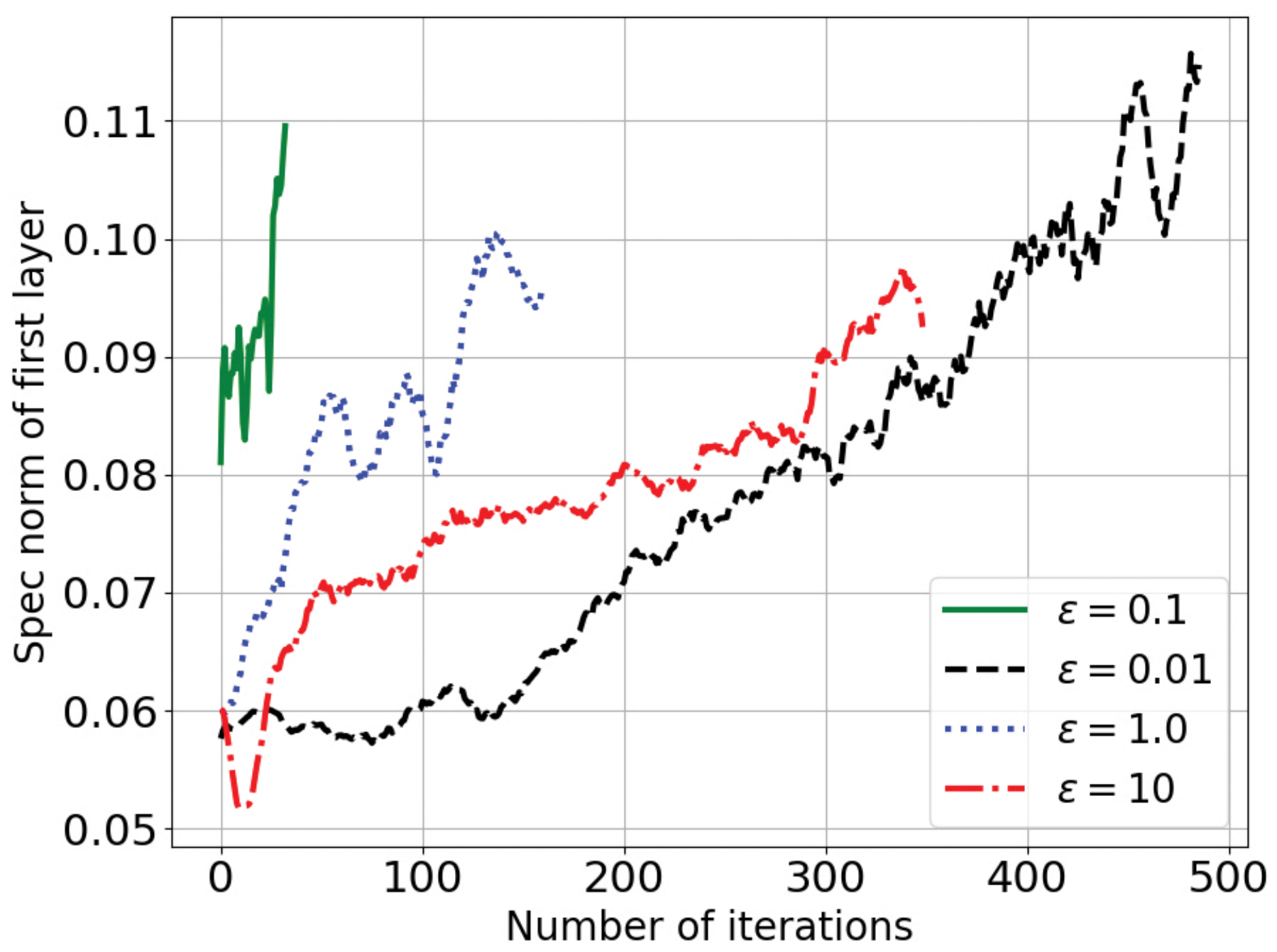}
  \caption{Approximated Spectral norm of gradients of first layer for different values of $\varepsilon$}\label{fig66}
\endminipage \quad
\minipage{0.3\textwidth}
  \includegraphics[width=\linewidth]{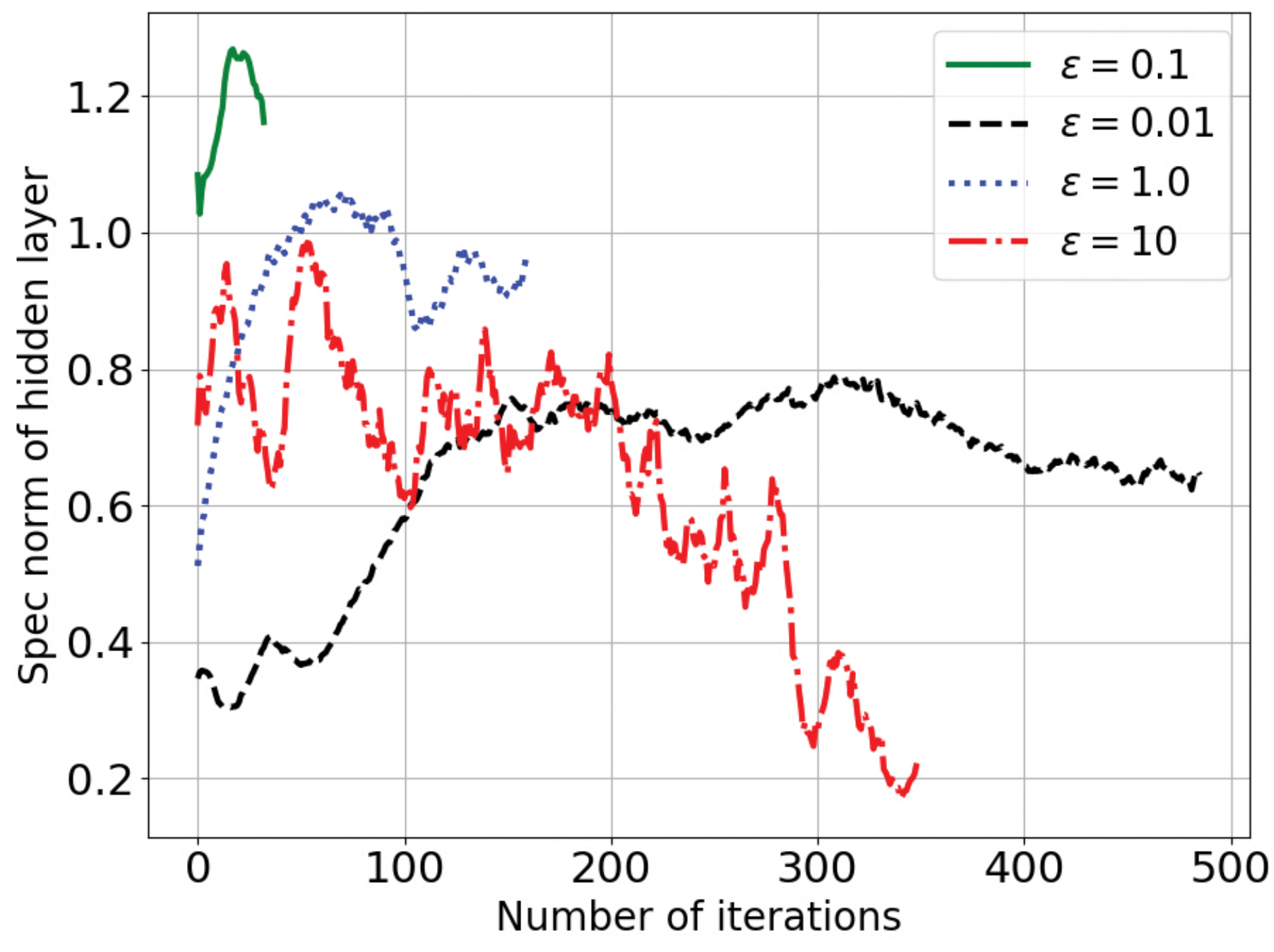}
  \caption{Approximated Spectral norm of gradients of hidden layer for different values of $\varepsilon$}\label{fig77}
  \endminipage
\end{figure}

\begin{figure}[!htb]
\centering
\minipage{0.32\textwidth}
  \includegraphics[width=\linewidth]{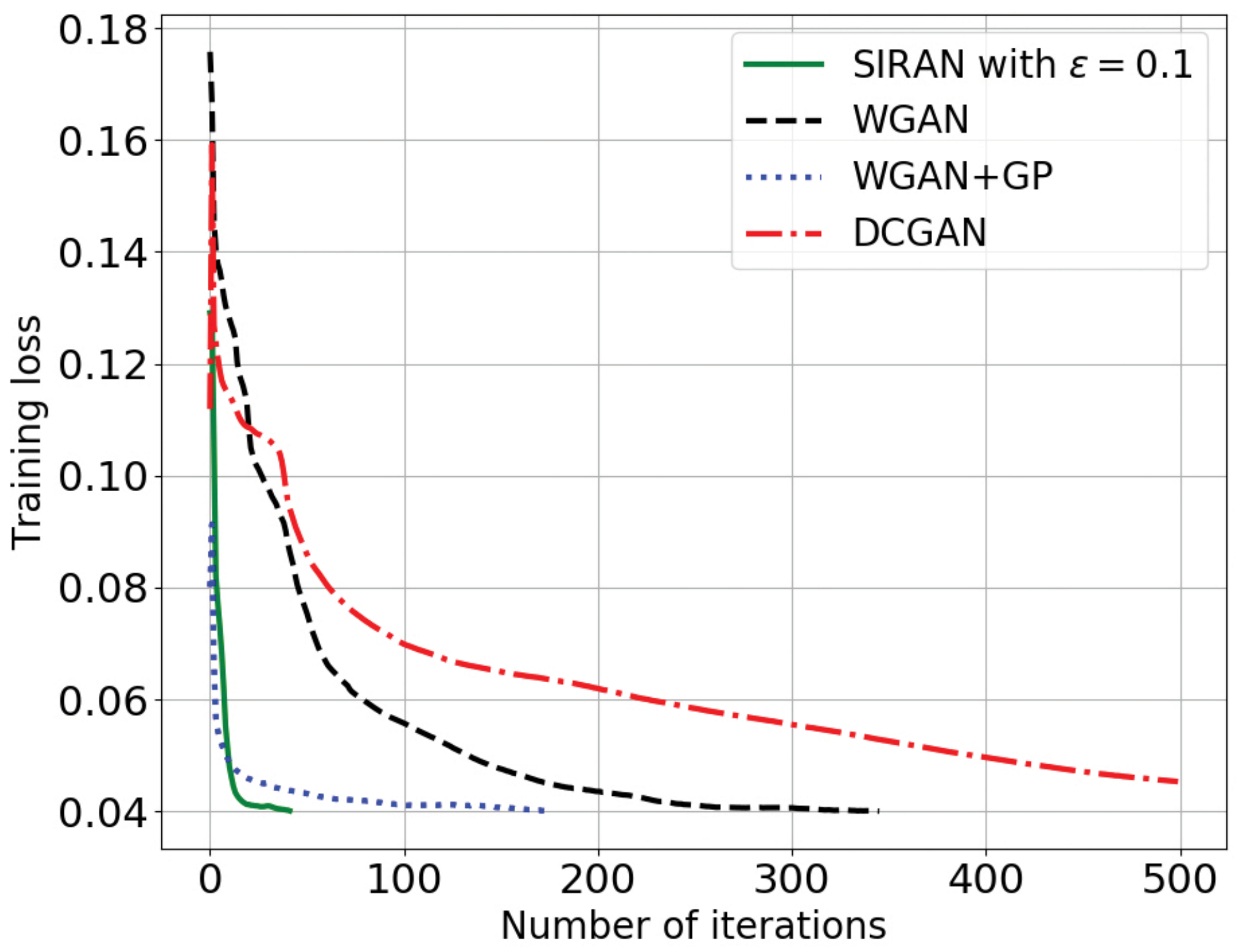}
  \caption{Training Loss for different learning methods}\label{fig555}
\endminipage \quad
\minipage{0.32\textwidth}
  \includegraphics[width=\linewidth]{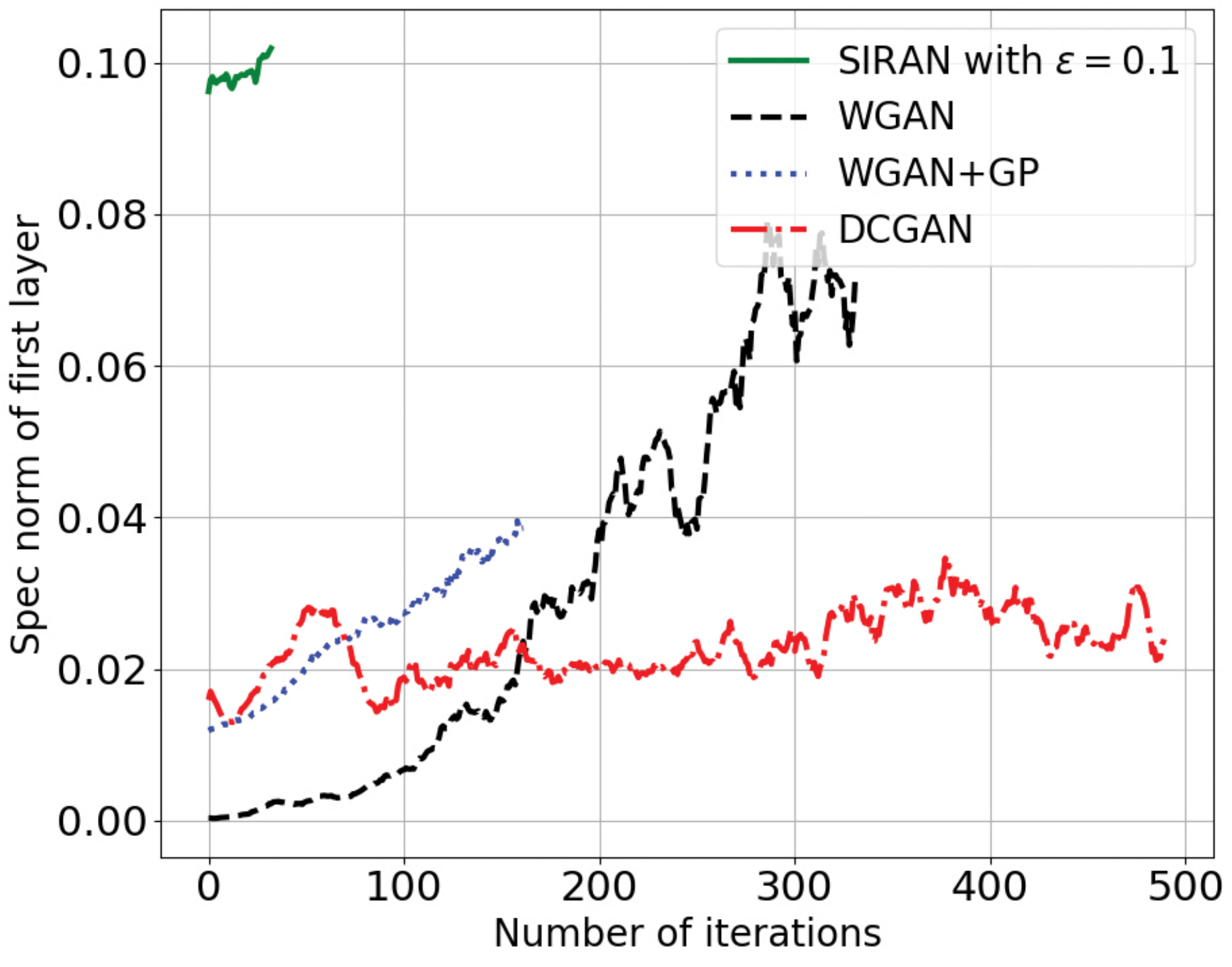}
  \caption{Approximated Spectral norm of gradients of first layer for different learning methods}\label{fig666}
\endminipage \quad
\minipage{0.3\textwidth}
  \includegraphics[width=\linewidth]{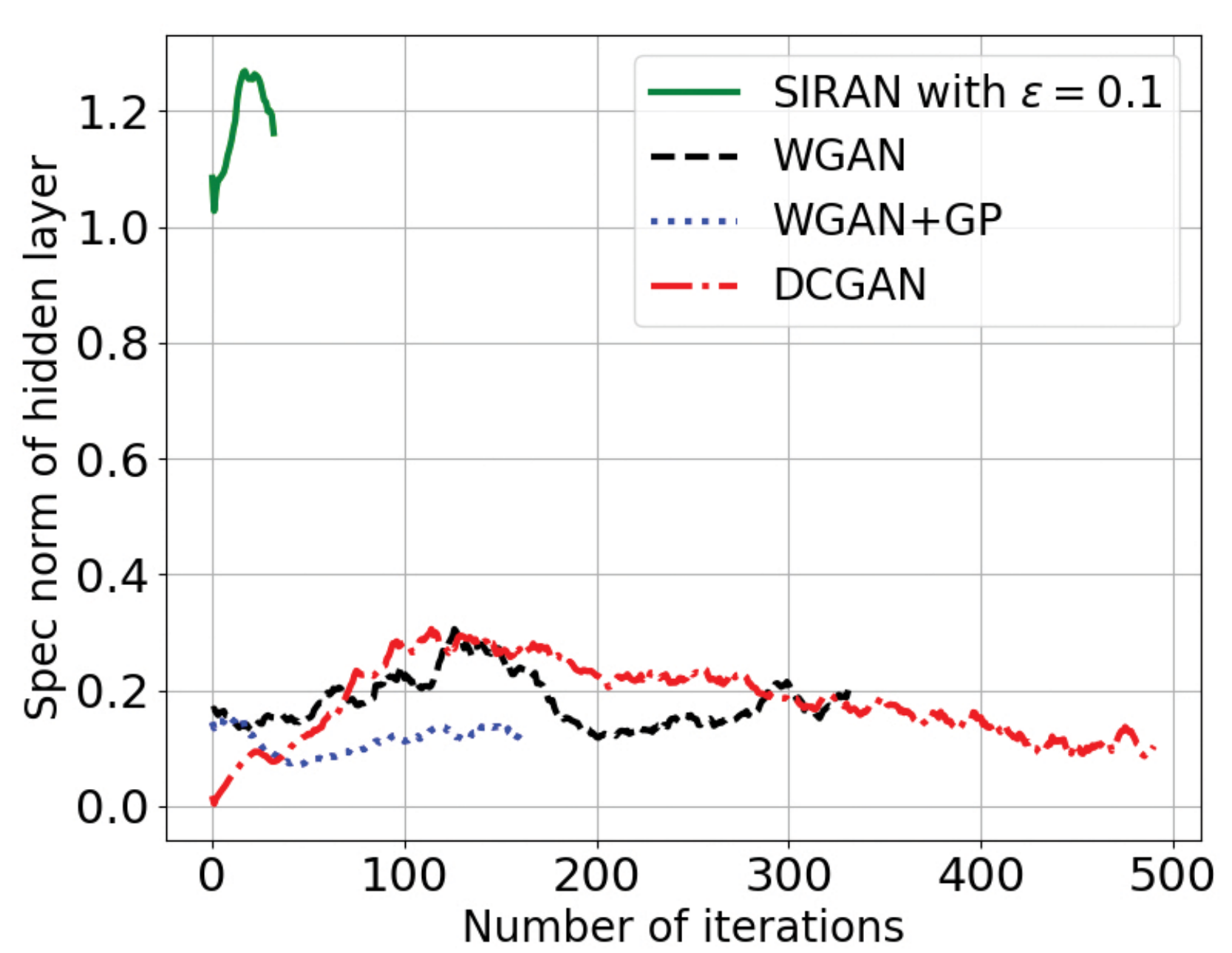}
  \caption{Approximated Spectral norm of gradients of hidden layer for different learning methods}\label{fig777}
  \endminipage
\end{figure}

\end{document}